\documentclass[11pt]{article}

\usepackage[utf8]{inputenc}
\usepackage[T1]{fontenc}
\usepackage[english]{babel}
\usepackage{lmodern}
\usepackage{showlabels,soul}

\usepackage{xcolor}

\usepackage[colorlinks=False]{hyperref}  

\usepackage{comment}

\usepackage{amsmath}  
\usepackage{amsfonts,dsfont}
\usepackage{amssymb}
\usepackage{amsthm}
\usepackage{mathtools}
\usepackage{thmtools}
\usepackage{bm} 

\numberwithin{equation}{section} 

\newlength{\bredde}
\def\slash#1{\settowidth{\bredde}{$#1$}\ifmmode\,\raisebox{.15ex}{/}
\hspace*{-\bredde} #1\else$\,\raisebox{.15ex}{/}\hspace*{-\bredde} #1$\fi}
\DeclarePairedDelimiter{\parentheses}{\lparen}{\rparen} 
\DeclarePairedDelimiter{\brackets}{\lbrack}{\rbrack} 
\DeclarePairedDelimiter{\absolute}{\lvert}{\rvert} 
\DeclarePairedDelimiter{\norm}{\lVert}{\rVert} 
\DeclarePairedDelimiter{\angles}{\langle}{\rangle} 
\DeclarePairedDelimiter{\floor}{\lfloor}{\rfloor} 

\renewcommand{\epsilon}{\varepsilon} 
\let\Re\relax\DeclareMathOperator{\Re}{Re} 
\let\Im\relax\DeclareMathOperator{\Im}{Im} 

\newcommand*{\defequals}{\vcentcolon=} 
\newcommand*{\euler}{\mathrm{e}} 
\newcommand*{\iunit}{\mathrm{i}} 
\newcommand*{\dif}{\mathop{}\!\mathrm{d}} 

\newcommand*{\Nset}{\mathds{N}} 
\newcommand*{\Rset}{\mathds{R}} 
\newcommand*{\Cset}{\mathds{C}} 

\newcommand{\be}{\begin{equation}}
\newcommand{\ee}{\end{equation}}

\usepackage{extpfeil}
\newextarrow{\xrightrightarrows}{{5}{8}{0}{0}}
{\bigRelbar\bigRelbar{\bigtwoarrowsleft\rightarrow\rightarrow}}

\DeclareMathOperator{\erf}{erf} 
\DeclareMathOperator{\Pf}{Pf} 
\newcommand*{\p}[1]{\parentheses*{#1}} 
\newcommand*{\bk}[1]{\brackets*{#1}} 
\newcommand*{\abs}[1]{\absolute*{#1}} 
\newcommand*{\cconj}[1]{\overline{#1}} 
\newcommand*{\sfrac}[2]{{#1}/{#2}} 

\declaretheorem[numberwithin=section]{proposition}
\declaretheorem[numberlike=proposition]{theorem}
\declaretheorem[numberlike=proposition]{lemma}
\declaretheorem[numberlike=proposition]{corollary}
\declaretheorem[numberlike=proposition]{remark}
\declaretheorem[numberlike=proposition]{example}
\declaretheorem[numberlike=proposition]{definition}

\DeclareMathOperator{\jpdf}{P} 
\DeclareMathOperator{\corrfct}{R}
\DeclareMathOperator{\kernel}{K}
\DeclareMathOperator{\prekernel}{\sigma} 
\newcommand*{\scp}[2]{\angles*{{#1}, {#2}}} 
\newcommand*{\skp}[2]{\angles*{{#1}, {#2}}_{\text{s}}} 

\makeatletter
\newcommand*{\doublerightarrow}[2]{\mathrel{
  \settowidth{\@tempdima}{$\scriptstyle#1$}
  \settowidth{\@tempdimb}{$\scriptstyle#2$}
  \ifdim\@tempdimb>\@tempdima \@tempdima=\@tempdimb\fi
  \mathop{\vcenter{
    \offinterlineskip\ialign{\hbox to\dimexpr\@tempdima+1em{##}\cr
    \rightarrowfill\cr\noalign{\kern.5ex}
    \rightarrowfill\cr}}}\limits^{\!#1}_{\!#2}}}
\makeatother

\begin{document}
\title{Skew-orthogonal polynomials in the complex plane and their Bergman-like kernels
}
\author{{\sc Gernot Akemann$^1$, Markus Ebke$^2$, and Iv\'an Parra$^3$}\\~\\
$^1$Faculty of Physics and $^2$Faculty of Mathematics, Bielefeld University,\\ PO-Box 100131, D-33501 Bielefeld, Germany\\
akemann@physik.uni-bielefeld.de, markus.ebke@uni-bielefeld.de\\$^3$Department of Mathematics, Katholieke Universiteit Leuven, \\Celestijnenlaan 200B box 2400, BE-3001 Leuven, Belgium\\
ivan.parra@kuleuven.be}

\date{}

\maketitle

\begin{abstract}

Non-Hermitian random matrices with symplectic symmetry provide examples for Pfaffian point processes in the complex plane. These point processes are characterised by a matrix valued kernel of skew-orthogonal polynomials. We develop their theory in providing an explicit construction of skew-orthogonal polynomials in terms of orthogonal polynomials that satisfy a three-term recurrence relation, for general weight functions in the complex plane. New examples for symplectic ensembles are provided, based on recent developments in orthogonal polynomials on planar domains or curves in the complex plane. Furthermore, Bergman-like kernels of skew-orthogonal Hermite and Laguerre polynomials are derived, from which the conjectured universality of the elliptic symplectic Ginibre ensemble and its chiral partner follow in the limit of strong non-Hermiticity at the origin. A Christoffel perturbation of skew-orthogonal polynomials as it appears in applications to quantum field theory is provided.

\end{abstract}

\section{Introduction}

The study of orthogonal and skew-orthogonal polynomials in the complex plane is closely related to the question of integrability of determinantal and Pfaffian point processes in the plane. Many of the known examples for such point processes can be realised as complex eigenvalues of non-Hermitian random matrices, such as the three Ginibre ensembles of Gaussian random matrices with real, complex or quaternion matrix elements \cite{Ginibre,LehmannSommers,Edelman}, or their three chiral counterparts \cite{James,A05,APS}. Compactly supported examples include the truncation of random orthogonal \cite{KSZ}, unitary \cite{ZS} and symplectic matrices \cite{BL}, and we refer to \cite{KS} for a review on non-Hermitian random matrices. 
At the same time these point processes are examples for two-dimensional Coulomb gases in a confining potential with specific background charge, and we refer to \cite{Peter} for details. A further example in this class is the circular quaternion  ensemble of random matrices belonging to the symplectic group $\mathbb{S}p(2N)$, distributed according to Haar measure, see \cite{EMeckes} for details. Its joint density of  eigenvalues distributed on the unit circle represents a Pfaffian point process that is also determinantal.

Apart from this direct statistical mechanics interpretation as a Coulomb gas, further applications of these point processes include dissipative quantum maps \cite{Haake}, dynamical aspects of neural networks \cite{Crisanti}, properties of the quantum Hall effect \cite{PdF} and quantum field theories with chemical potential \cite{James,A05,APS}. In quantum optics the Bergman kernel of planar Hermite polynomials plays an important role in the construction of coherent and squeezed states \cite{ABG}.
Notably, in the symplectic symmetry class, that will be our focus here, there is a map from the symplectic Ginibre ensemble to disordered non-Hermitian Hamiltonians in an imaginary magnetic field \cite{EK}. 
The circular quaternion  ensemble finds applications in the computation of thermal conduction in superconducting quantum dots \cite{DBB}.  
The predictions of the non-Hermitian symplectic chiral symmetry class \cite{A05} were successfully compared with data from lattice simulations in Quantum Chromodynamics with two colours at non-vanishing chemical potential \cite{ABi}. Here, the insertion of quark flavours, given in terms of characteristic polynomials in the random matrix setting, play an important role and were tested in \cite{ABi}. We will investigate the effect of these insertions on the underlying skew-orthogonal polynomials (SOP) also known as Christoffel perturbation.

The real and symplectic Ginibre ensemble differ from the complex ensemble in the following way. First, in the two former ensembles eigenvalues come in complex conjugated pairs. Second, the real ensemble shows an accumulation of eigenvalues on the real axis, while the symplectic ensemble shows a depletion, as the probability to have real eigenvalues is zero. 
Let us briefly recall what is known about the symplectic Ginibre ensemble in the limit of large matrix size (or number of particles). 
It is not surprising that the gap probability \cite{Mehta} and all correlation functions at the origin \cite{Mehta,Kanzieper} differ from the complex Ginibre ensemble. The same statement holds for the density at weak non-Hermiticity \cite{EK}. Below, we will present a proof that in the limit of strong non-Hermiticity the correlation functions at the origin of the symplectic Ginibre ensemble are universal, in the sense that they hold for the elliptic ensemble beyond the rotationally invariant case. This was conjectured in \cite{EK}, and the same conjecture \cite{A05} for the chiral ensemble will be shown as well. 

In contrast, the distribution of the largest eigenvalue in radius \cite{Rider} as well as the local radial density in the bulk away from the real line agree with those of the complex Ginibre ensemble \cite{Jesper}. Without integration over the angles this agreement has been shown more recently \cite{AKMP} for all correlation functions (marginals) in the bulk of the spectrum of the symplectic Ginibre ensemble, away from the real axis. For the same agreement between the real and complex Ginibre ensemble we refer to \cite{BS}. This is in strong contrast with the Hermitian ensembles of random matrices, corresponding to a so-called Dyson gas on the real line at inverse temperature $\beta=1,2,4$. The local statistics for the latter three ensembles differs everywhere in the spectrum,  see e.g. \cite{Mehta,Peter} for a summary of results.
It is one of the goals of this article to develop the theory of SOP in the complex plane. This will allow us to construct further examples of Pfaffian point processes with symplectic symmetry that are integrable, where this extended universality in the complex plane can be studied. 

What is known for the construction of SOP for general ensembles with symplectic symmetry? In \cite{Kanzieper} it was shown that both polynomials of odd and even degree enjoy a Heine-like representation. It is given as the expectation value of a characteristic polynomial (times a trace in the case of polynomials of odd degree), that is a multiple integral representation of the order of the degree of the SOP. Only in a limited number of cases have these been used for an explicit construction, using Schur polynomials \cite{Petersymplectic} or Grassmann integrals \cite{APS}. A second construction sets up a Gram-Schmidt skew-orthogonalisation procedure \cite{AKP} that we will recall below. However, also this is of limited use for an explicit construction. 
Current explicit examples for SOP include Hermite \cite{Kanzieper} and Laguerre polynomials \cite{A05} for the elliptic symplectic Ginibre ensemble and its chiral counterpart, respectively.

In this article we will exploit the particularly simple structure of the skew-product in symplectic ensembles. It allows to relate the skew-product to the multiplication acting on the standard Hermitian inner product. 
It is well known that on subsets of the real line orthogonal polynomials (OP) always satisfy a three-term recurrence relation with respect to multiplication. In the plane this is no longer true, and on bounded domains we may not expect any finite-term recurrence in general \cite{Lempert}. It was shown much later on bounded domains with flat measure that if such a recurrence exists, the domain is an elliptic disc and the depth of recurrence is three, see \cite{Khavinson} and references therein. However, in the weighted case the ellipse is no longer special, as also here OP without recurrence exist \cite{ANPV}. On the other hand, if we can promote classical OP from the real line to the complex plane, then the existence of such a recurrence is always guaranteed.

The remainder of this article is organised as follows. In the next Section \ref{sec:setup} we define the class of ensembles of complex eigenvalues with symplectic symmetry we consider here, including their skew-product. 
We recall the Gram-Schmidt skew-orthogonalisation that leads to the reproducing polynomial kernel, in terms of which all $k$-point correlation functions or marginals are given. 
In Section \ref{sec:construct} an 
explicit construction is provided for SOP for a general class of weight functions.  
They are characterised by OP that satisfy a three-term recurrence relation with real coefficients, orthogonal with respect to the same real valued weight function. 
Several new examples of resulting planar SOP are presented, including a weight of Mittag-Leffler type in the plane and weights on an elliptic disc that lead to planar Gegenbauer SOP, as well as a subfamily of non-symmetric Jacobi polynomials that include Chebyshev polynomials. A realisation of the Chebyshev polynomials as 
Szeg\H{o} SOP on an ellipse is given, too. In Appendix \ref{appA} we recover some known planar SOP from our construction.
Section \ref{sec:PK-sop} is devoted to the derivation of the Bergman-like kernel for Hermite and Laguerre type SOP. These formulas imply the large-$N$ limit at the origin of the spectrum for all correlation functions at strong non-Hermiticity   
in the corresponding ensembles. This confirms their conjectured universality within these two elliptic classes.
In Section \ref{sec:CP} a Christoffel perturbation of SOP is provided for general weight functions.
As a corollary we obtain that a Christoffel perturbation does not preserve the three-term recurrence relation for measures on $ \mathds{C} $. 
The corresponding Fourier coefficients are provided in Appendix \ref{appSOP1}. Appendix \ref{appC} contains a collection of integrals needed throughout the article.

\section*{Acknowledgments}
Funded by the Deutsche Forschungsgemeinschaft (DFG, German Research Foundation) - SFB 1283/2 2021 - 317210226 
``Taming uncertainty and profiting from randomness and low regularity in analysis, stochastics and their applications'' (G.A. and I.P.) and 
IRTG2235 ``Searching for the regular in the irregular: Analysis of singular and random systems'' (M.E.), 
and by the grants DAAD-CONICYT/Becas Chile, 2016/91609937 and 
FWO research grant G.0910.20 
(I.P.). 
We thank  Boris Khoruzhenko for useful discussions and Sung-Soo Byun for discussions and many useful comments on this manuscript.

\section{Symplectic ensembles and skew-orthogonal polynomials} \label{sec:setup}

The point processes on a subset of the complex plane with symplectic symmetry considered in this paper are defined by the following joint probability density
\begin{equation} \label{eq:eigenvalue_distribution}
	\dif\! \jpdf_N(z_1, \dots, z_N)
	\defequals \frac{1}{Z_N}
		\prod_{k > l}^{N} \abs{z_k - z_l}^2 \abs{z_k - \cconj{z_l}}^2
		\prod_{j= 1}^{N} \abs{z_j - \cconj{z_j}}^2 \prod_{i = 1}^{N} \dif\mu(z_i).
\end{equation}
Here, $\mu$ is a positive Borel measure on the domain $D$ in the complex plane. 
Further conditions on the measure will be specified at the beginning of Subsection \ref{subsec:sop}.
The normalization constant $Z_N$ (partition function) is given by
\begin{equation}\label{skewpartitionfunction}
	Z_N \defequals \int \prod_{k > l}^{N} \abs{z_k - z_l}^2 \abs{z_k - \cconj{z_l}}^2
		\prod_{j = 1}^{N} \abs{z_j - \cconj{z_j}}^2\ \prod_{i = 1}^{N} \dif\mu(z_i)>0,
\end{equation}
The joint density \eqref{eq:eigenvalue_distribution} may result for example from the distribution of the $2N$ eigenvalues  
$(z_1, \cconj{z_1}, \dots,$ $ z_N, \cconj{z_N})$ of an $N \times N$ dimensional quaternionic non-Hermitian random matrix (or its $2N$ dimensional complex representation).
Examples for such random matrix realisations include the elliptic quaternionic Ginibre ensemble \cite{Kanzieper} and its chiral counterpart \cite{A05}. When the measure $\mu$ is supported on the unit circle, further representatives include the circular quaternion ensemble with a flat measure, see \cite[Thm. 3.1]{EMeckes}. It is not difficult to see, that for $z_k=\euler^{\iunit\theta_k}$ the first product in \eqref{eq:eigenvalue_distribution} leads to the Vandermonde determinant squared in the variables $x_k=\cos{\theta_k}$, and thus to a determinantal point process. The second product in \eqref{eq:eigenvalue_distribution}, $\prod_{k=1}^N4(1-x_k^2)$, is then taken as part of the weight function on $(-1,1)$. In the same way also the circular real ensembles of the Haar distributed groups $\mathbb{SO}(N)$ for even and odd $N$ provide examples for this Pfaffian point process \cite[Thm. 3.1]{EMeckes}, with varying weight functions though. All of these are also determinantal. 

We will be more general here in taking the complex eigenvalue model \eqref{eq:eigenvalue_distribution} as a starting point. 
In general, the ensemble defined in \eqref{eq:eigenvalue_distribution} is a Pfaffian point process, 
cf. \cite{Mehta,Kanzieper}. 
Defining the $k$-point correlation functions (or marginal measures) as
\begin{equation}\label{eq:Rkdef}
	\corrfct_{N, k}(z_1, \dots, z_k)
	\defequals \frac{N!}{(N - k)!}
	\int_{D^{N-k}}
		\jpdf_N(z_1, \dots, z_N)
	\prod_{i=k+1}^N\dif\mu(z_i),
\end{equation}
they can be shown to take the form \cite{Kanzieper} 
\begin{equation*}\label{eq:RkPfmu}
	\dif\!\corrfct_{N,k}(z_1, \dots, z_k)
	= \underset{1\leq i,j \leq k}{\Pf} \brackets[\Big]{\widehat{\kernel}_N(z_i, z_j)} \prod_{i=1}^k(\cconj{z_i} - z_i)\dif\mu(z_i).
\end{equation*}
Here, $\widehat{\kernel}_N$ is the $2 \times 2$ matrix-valued kernel of the corresponding point process.
Equivalently, when $\mu$ has the density $w$ with respect to the volume element $\dif\mathrm{v}$ on $D$, we have
\begin{equation}\label{eq:RkPf}
	\dif\!\corrfct_{N, k}(z_1, \dots, z_k)
	= \underset{1\leq i,j \leq k}{\Pf}\brackets[\Big]{\kernel_N(z_i, z_j)} \prod_{i=1}^k(\cconj{z_i} - z_i) \dif\mathrm{v}(z_i),
\end{equation}
where the $2 \times 2$ matrix-valued kernel $\kernel_N$ is defined as
\begin{equation}\label{eq:Kerneldef}
	\kernel_N(z, u)
	\defequals \sqrt{w(z) w(u)}
	\begin{pmatrix}
		\prekernel_N(z, u) & \prekernel_N(z, \cconj{u}) \\
		\prekernel_N(\cconj{z}, u) & \prekernel_N(\cconj{z}, \cconj{u})
	\end{pmatrix}.
\end{equation}
In that case we have $\kernel_N(z, u)=\sqrt{w(z) w(u)}\ \widehat{\kernel}_N(z,u)$. 
The function $\prekernel_N(z, u)$ 
is called the pre-kernel or \emph{polynomial kernel}, to be defined 
in \eqref{eq:prek-def} in terms of SOP.
Compared with \cite{Kanzieper} we have taken the pre-factors $(\cconj{z} - z) (\cconj{u} - u)$ out of the Pfaffian, to avoid cuts for the square root of these factors.
As an example we get for the one-point and two-point function 
\begin{align}\label{eq:R1}
	\corrfct_{N, 1}(z) &= \p{\cconj{z} - z} w(z) \prekernel_N(z, \cconj{z}), \\
	\corrfct_{N, 2}(z_1, z_2) &= \p{\cconj{z_1} - z_1} \p{\cconj{z_2} - z_2}
		w(z_1) w(z_2) 
		\bk{
			\prekernel_N(z_1, \cconj{z_1}) \prekernel_N(z_2, \cconj{z_2})
			- \abs{\prekernel_N(z_1, z_2)}^2
			+ \abs{\prekernel_N(z_1, \cconj{z_2})}^2
		}.
		\nonumber
\end{align}

The pre-kernel can be expressed in terms of SOP to be defined below, being a particularly simple choice in a more general construction, see \cite{Kanzieper}. There, it is given in terms of the inverse moment matrix, where the de~Bruijn integral formula is applied together with the fact that the joint density \eqref{eq:eigenvalue_distribution} is proportional to the product 
$\prod_{j=1}^N(\cconj{z}_j-z_j)$  times the Vandermonde determinant of size $2N$ of all eigenvalues and their complex conjugates.

\subsection{Skew-orthogonal polynomials}\label{subsec:sop}

From now on let $\mu$ be a positive Borel measure on $\Cset$, with an infinite number of points in its support $D$, and such that $\int \abs{z}^m \dif\mu(z)<\infty$ for all non-negative integers $m$.
Further we require that $D$ is symmetric about the real axis, i.e.~$z\in D$ if and only if $\cconj{z} \in D$.
For any $f,g \in \Cset[z]$, we define the following skew-symmetric form 
\begin{equation}\label{eq:skewproddef}
	\skp{f}{g} \defequals
		\int
		\p{f(z) g(\cconj{z}) - g(z) f(\cconj{z})}
		\p{z - \cconj{z}}  \dif\mu(z).
\end{equation} 
Equivalently, $\skp{\cdot}{\cdot}$ is an alternating form. In particular, when the polynomials $f, g$ have real coefficients, $\skp{\cdot}{\cdot}$ is also a skew-Hermitian form.

\begin{definition}\label{def:sopdef}
A family of polynomials $\p{q_n}_{n \in \Nset}$ with $\deg q_n = n$ is called \emph{skew-orthogonal} corresponding to $\mu$, if they satisfy for all non-negative integers $k, l \in \Nset$:
\begin{align}
	&\skp{q_{2 k}}{q_{2 l}} = \skp{q_{2 k + 1}}{q_{2 l + 1}} = 0, \label{eq:SOPdef1} \\
	&\skp{q_{2 k}}{q_{2 l + 1}} = -\skp{q_{2 l + 1}}{q_{2 k}} = r_k \delta_{k, l}, \label{eq:SOPdef2}
\end{align}
with $r_k > 0$ being their skew-norms.
\end{definition}

The SOP are called \emph{monic} if their leading coefficient is unity, i.e.~$q_n(z) = z^n + \mathcal{O}(z^{n - 1})$. 
Notice that this choice of the leading coefficient does not make the SOP unique.
The reason is that the odd polynomials $q_{2 n + 1}$ can be modified by adding an arbitrary multiple of the even polynomial $q_{2 n}$, $\tilde{q}_{2 n + 1}(z)=q_{2 n + 1}(z)+d_nq_{2n}(z)$, without changing the skew-orthogonality relations \eqref{eq:SOPdef1} and \eqref{eq:SOPdef2}. Is it  not difficult to see, using the Heine-like representation of the SOPs given in terms of $2n-$fold integrals in \cite{Kanzieper}, that different transformations do not preserve the skew orthogonality conditions \eqref{eq:SOPdef1} and \eqref{eq:SOPdef2}. Hence, by imposing the SOP to be monic and the coefficient of $z^{2n}$ in $q_{2n+1}(z)$ to be zero, 
we fix this ambiguity and make them unique, see \autoref{uniquesop} below.
However, sometimes it is convenient to 
choose the coefficient of $z^{2n}$ to be non-vanishing, in order 
to obtain  closed formulas for the odd SOP, see Section \ref{sec:construct} for various examples. The existence of the SOP is guaranteed by Gram-Schmidt skew-orthogonalisation to be provided below in \autoref{skewevenGramTHM}. As a consequence of Gram-Schmidt skew-orthogonalisation and the definition of the skew symmetric form $\skp{\cdot}{\cdot}$, the SOP belong to the polynomial ring $\mathds{R}[z]$. As usual, we denote by $\mathds{R}_{n}[z]$ the space of polynomials with real coefficients and  degree at most $n$ .

\begin{lemma}\label{uniquesop}
Let $\p{q_n}_{n \in \Nset}$ be a sequence of monic skew orthogonal polynomials, such that the coefficient $z^{2n}$ in $q_{2n+1}(z)$ is zero.  Then, the sequence of SOPs in \autoref{def:sopdef} is unique.
\end{lemma}

\begin{proof}
Suppose the existence of two monic SOPs of even degree $q_{2n},\tilde{q}_{2n}\; \in \mathds{R}_{2n}[z]$. Let $q:=q_{2n}-\tilde{q}_{2n} \in \mathds{R}_{2n-1}[z]$. Hence $\skp{q_{2n}}{zq}=\skp{\tilde{q}_{2n}}{zq}=0$ (this follows by expanding $zq$ either in terms of $q_{\ell}$ or $\tilde{q}_{\ell}$, and by using the skew orthogonality conditions \eqref{eq:SOPdef1} and \eqref{eq:SOPdef2}). But this implies that $0=\skp{q}{zq}=\overline{\skp{q}{zq}}$, from which we deduce that $
\int |q(z)|^2|z-\overline{z}|^2d\mu(z)=0$. Hence $q = 0$ $\mu-$almost everywhere, and because $\mu$ contains an infinite number of points in its support, we conclude that $q \equiv0$.

Now, suppose the existence of two monic SOPs of odd degree $q_{2n+1},\tilde{q}_{2n+1}\; \in \mathds{R}_{2n+1}[z]$. Let $p:=q_{2n+1}-\tilde{q}_{2n+1} \in \mathds{R}_{2n-1}[z]$. Since $q_{2n}=\tilde{q}_{2n}$, we obtain $-\tilde{r}_n=\skp{\tilde{q}_{2n+1}}{\tilde{q}_{2n}}=\skp{\tilde{q}_{2n+1}}{q_{2n}}=-r_n$. The last equality follows after expanding the $\tilde{q}_{2n+1}$ in terms of the $q_{\ell}$, and using the fact that the $\tilde{q}_{2n+1}$ are monic polynomials. Again, since $q_{2n}=\tilde{q}_{2n}$, we have for the Fourier expansion of $zp(z)$,

\be
z\,p(z)=\lambda_{2n}q_{2n}(z)+\sum_{\ell=0}^{2n-1}\lambda_{\ell}q_{\ell}(z)=\lambda_{2n}q_{2n}(z)+\sum_{\ell=0}^{2n-1}\tilde{\lambda}_{\ell}\tilde{q}_{\ell}(z).
\ee
Thus, $\skp{q_{2n+1}}{zp}=\skp{\tilde{q}_{2n+1}}{zp}=-\lambda_{2n}r_{n}$, therefore $0=\skp{p}{zp}=\overline{\skp{p}{zp}}$, which concludes the proof.
\end{proof}

\begin{remark}
When the measure $\mu$ has density $w$ on some domain $D$, then we talk about  
SOP with respect to the weight function $w$.
\end{remark}

A closed expression can be obtained for these SOP in terms of the \emph{skew-complex moments} of $\mu$
\begin{equation}
g_{i,j} \defequals \skp{z^i}{z^j}
	=\int
		\p{z^i \cconj{z}^j - z^j \cconj{z}^i}
		\p{z - \cconj{z}} \dif\mu(z)
	\in \Rset
\end{equation}
and the \emph{real skew-symmetric Gram matrix}
\begin{equation} \label{skewGram}
	G_k = 
	\begin{pmatrix}
		0 & g_{0,1} & \cdots & g_{0,2k-1} \\
		-g_{0,1} & 0 & \cdots & g_{1,2k-1} \\
		\vdots & \vdots & \ddots & \vdots \\
		-g_{0,2k-1} & -g_{1,2k-1} & \cdots & 0
	\end{pmatrix}_{2k \times 2k}.
\end{equation} 
De~Bruijn's integration formula implies that the Pfaffian of the skew-symmetric Gram matrix satisfies
\begin{equation}
\Delta_{-1} \defequals 1, \quad
\Delta_{k} \defequals \Pf G_{k+1} = \frac{1}{(k+1)!} \,Z_{k+1}\,>0, \quad \forall \,k \in \Nset,
\end{equation} 
with $Z_k$ defined in \eqref{skewpartitionfunction}. In analogy with the OP, the skew-complex moments already determine the SOP. In terms of these one can give explicit Pfaffian formulae
for the orthonormalised polynomials.

\begin{theorem}[Gram-Schmidt skew-orthogonalisation]\label{skewevenGramTHM} The \emph{skew-orthonormal} polynomials of even degree 
$\hat{q}_{2k}$ are formed by replacing in $G_{k+1}$ the $(2k+2)$nd row and column by powers of $z$,
\begin{equation} \label{skewevenGram}
	{\hat q}_{2k}(z) = \frac{1}{\sqrt{\Delta_{k}\Delta_{k-1}}} \Pf 
	\begin{pmatrix}
		0 & g_{0,1} & \cdots & g_{0,2k} & 1 \\
		-g_{0,1} & 0 & \cdots & g_{1,2k} & z \\
		\vdots & \vdots & \ddots & \vdots & \vdots \\
		-g_{0,2k} & -g_{1,2k} & \cdots & 0 & z^{2k} \\
		-1 & -z & \cdots & -z^{2k} & 0
	\end{pmatrix},
\end{equation}
and for the odd degree 
${\hat q}_{2k+1}$, by replacing the $(2k+2)$nd row and column by powers of $z$ except $z^{2k}$, as well as shifting the index of the Gram matrix in the $(2k+1)$st row and column by one, 
\begin{eqnarray} \label{skeweoddGram}
	{\hat q}_{2k+1}(z) = \frac{1}{\sqrt{\Delta_{k}\Delta_{k-1}}} \Pf 
	\begin{pmatrix}
		0 & g_{0,1} & \cdots &g_{0,2k-1}& g_{0,2k+1} & 1 \\
		-g_{0,1} & 0 & \cdots &g_{1,2k-1}& g_{1,2k+1} & z \\
		\vdots & \vdots & \ddots & \vdots & \vdots & \vdots \\
		-g_{0,2k-1}&-g_{1,2k-1}&\cdots &0&g_{2k-1,2k+1}& z^{2k-1}\\
		-g_{0,2k+1} & -g_{1,2k+1} & \cdots &-g_{2k-1,2k+1}& 0&z^{2k+1}\\
				-1 & -z & \cdots &-z^{2k-1}& -z^{2k+1} &0 
	\end{pmatrix}.	\qquad
\end{eqnarray} 
\end{theorem}
In particular this implies that the SOP have real coefficients.
The proof of this theorem involves properties of Pfaffians. It was presented first in \cite{AHM} for the case of SOP over the real line\footnote{Here, the skew-symmetric product is $\angles{f, g}_4=\frac{1}{2}\int (f(x)g'(x)-f'(x)g(x))\dif\mu(x)$.}, and extended to planar SOP in \cite{AKP}. Here, we emphasise that this construction holds without specifying the support of the measure $\mu$.

\begin{remark}
We stress that the SOP always exist, under the minimal assumptions on the measure that we stated at the beginning of this subsection, even without the support being symmetric about the real axis.

Also note that, by using de~Bruijn's formula, we can easily express the partition function \eqref{skewpartitionfunction} in terms of the skew-norms from \eqref{eq:SOPdef2}:
\begin{equation} \label{eq:ZNprod}
	Z_N = N! \prod_{k=0}^{N-1} r_k,
\end{equation}
see also \cite[Sec.~15.2]{Mehta} for the symplectic Ginibre case.
Given the positivity of $Z_N$ for all $N$ this automatically implies the positivity of the skew-norms $r_k$. Conversely, all examples in Section \ref{sec:construct} where we determine these skew-norms explicitly for a given measure provide instances of Selberg-like integrals.
\end{remark}

The representation of the SOP in terms of Gram matrices may be important from a theoretical point of view, but is not very useful  for the actual computation of the SOP since it involves the evaluation of Pfaffians. In Section \ref{sec:construct} we will propose a more explicit construction, given an orthonormal system in $L^2(\dif\mu)$ that satisfies certain properties.

Let $\mathcal{P}_{2n}$ be the space of analytic polynomials of degree at most $2n-1$ 
(i.e. $p \in \mathcal{P}_{2n}$ with $\frac{\partial}{\partial\bar{z}}p=0$), 
and we equip $\mathcal{P}_{2n}$ with the skew-product $\skp{\cdot}{\cdot}$.

\begin{lemma}\label{lem:skewrepro} The polynomial kernel $\sigma_{n}(u,v)$ defined as follows
\begin{equation}\label{eq:prek-def}
	\sigma_n(u, v)
	\defequals \sum_{k = 0}^{n - 1} \frac{
			q_{2 k + 1}(u) q_{2 k}(v) - q_{2 k}(u) q_{2 k + 1}(v)
		}{r_k}, 
\end{equation}
is the reproducing skew-kernel of $(\mathcal{P}_{2n}, \skp{\cdot}{\cdot})$.
\end{lemma}

\begin{proof} Let $f \in \mathcal{P}_{2n}$ and $\sigma_{v}(u) \defequals \sigma_{n}(u,v)$. Without restriction we can assume that $f=q_{2m}$ or $f=q_{2m+1}$ with $m \leq n-1$, thus
\begin{align*}
\skp{f}{\sigma_v} &= \sum_{k = 0}^{n - 1}\frac{1}{r_k} \bk{q_{2k}(v) \skp{f}{q_{2k+1}} + q_{2k+1}(v) \skp{q_{2k}}{f}} \\
&= f(v).
\end{align*}
\end{proof}

\begin{remark}\label{Rem2.6}
\autoref{lem:skewrepro} tells us that $\sigma_n$ reproduces itself, $\skp{\sigma_v}{\sigma_u} = \sigma_v(u) = \sigma_n(u,v)$, and that $\sigma_n$ is a skew-symmetric function, $\sigma_n(u,v) = -\sigma_n(v,u)$. These two properties are in complete analogy to the Hermitian inner product space $L^2(\dif\mu)$ with scalar product $\scp{\cdot}{\cdot}$, where the polynomial kernel
\begin{equation}
K_v(u) \defequals K_n(u,v) = \sum_{k=0}^{n-1}P_{k}(u)\cconj{P_k(v)}
\end{equation} 
is expressed in terms of the orthonormal polynomials $P_n$ corresponding to $\mu$. It satisfies the reproducing property $\scp{K_v}{K_u} = K_n(u,v)$ and is a Hermitian function, i.e.~$K_n(v,u) = \cconj{K_n(u,v)}$.

The polynomial kernel $\sigma_{n}(u,v)$ is not affected by the non-uniqueness of the odd SOP. It remains unchanged if we redefine $\tilde{q}_{2 n + 1}(z)=q_{2 n + 1}(z)+d_nq_{2n}(z)$ as the latter terms drop out in \eqref{eq:prek-def}, due to anti-symmetry.
\end{remark}

\begin{remark} As was shown in \cite{Kanzieper} the pre-kernel $\sigma_n(u,v)$ is the main input of the $2 \times 2$ matrix-valued kernel $\kernel_n$ in \eqref{eq:Kerneldef}.
\end{remark}

\begin{remark}\label{pfaffcocycle} Let $c(u,v) = g(u)g(v)$, with $g$ a continuous unimodular function, such that $g\p{\cconj{u}} = 1/g(u)$. Then, the Pfaffian of a $2k \times 2k$ skew-symmetric matrix $(a_{i,j})_{1 \leq i,j \leq 2k}$ remains unchanged if we multiply each $a_{i,j}$ by $c(u_i, u_j)$, where the $u_1, \dots, u_{2k}$ come in complex conjugate pairs, $u_{k + i} = \cconj{u_i}$.
In particular, $\widetilde{\prekernel}(u, v) = g(u) g(v) \prekernel(u, v)$ is another pre-kernel giving rise to the same correlation functions $\corrfct_{N, k}$.
Two such pre-kernels as well as their corresponding kernels are then called \emph{equivalent}.
\end{remark}

\section{Construction of skew-orthogonal polynomials}\label{sec:construct}

In this section we reduce the construction of the SOP corresponding to $\mu$ to that of OP corresponding to the same measure $\mu$, that satisfy a (suitable) three-term recurrence relation. Let us emphasise that, in general, the construction of an orthonormal system in $L^2(\Cset,\dif\mu)$ is straightforward using Gram-Schmidt orthogonalisation, see \cite{Walter}. However, generic OP do \emph{not} satisfy such a recurrence, see \cite{Lempert} for a recent discussion. On the other hand all (polynomial) orthogonal systems in  $L^2(\Rset,\dif\nu)$ do satisfy such a three-term recurrence with real coefficients, and thus provide many potential candidates. If we can thus find a measure $\mu$ with support on the complex plane, such that  the same orthogonal system in  $L^2(\Rset,\dif\nu)$ gives rise to an orthogonal system in $L^2(\Cset,\dif\mu)$, our Theorem \ref{thm:sop_from_op} below immediately leads to SOP  corresponding to the same measure $\mu$. Below we will give several examples of such a situation. 

Let $\mu$ be a measure that satisfies the properties from the previous Subsection \ref{subsec:sop}, 
then by the Gram-Schmidt process one can construct a unique sequence of polynomials
\begin{equation}
p_n(z)=\gamma_n z^n+\mathcal{O}(z^{n-1}), \quad \gamma_n>0,
\end{equation}
that form an orthogonal system in $L^{2}(\dif\mu)$,
\begin{equation}\label{eq:proddef}
\scp{p_n}{p_m} =\int p_n(z) \cconj{p_m(z)} \dif\mu(z)=h_n\delta_{n,m}.
\end{equation}
Due to the linearity of the Hermitian form $\scp{\cdot}{\cdot}$ over $\Rset$ 
we can assume that $\gamma_n \equiv 1$. In this case we say that the sequence $(p_n)_{n \in \Nset}$ of OP is  chosen in \emph{monic normalisation} (i.e.~$p_n(z)=z^n+\dots$).
These OP (corresponding to $\mu$) satisfy a three-term recurrence relation with real coefficients, if
\begin{equation}\label{eq:3term}
	z p_k(z) =  
	p_{k + 1}(z) + b_k p_k(z) + c_k p_{k - 1}(z), \qquad
	b_k, c_k \in \Rset, \ k = 1, 2, \ldots
\end{equation}
The condition for $b_k$ and $c_k$ to be real is obviously equivalent to the condition that the planar OP have real coefficients, $\cconj{p_n(z)} = p_n(\cconj{z})$ for all $n\in \Nset$.

\begin{theorem} \label{thm:sop_from_op} Let $(\mu_{k,j})_{k,j \in \Nset}$ be a sequence of real numbers such that $\mu_{k, k} = 1$ and $\mu_{k,j}=\lambda_{k-1} \mu_{k-1,j}$, with $\lambda_{k-1}$ independent of $j$, for $j < k$. And let $(p_n)_{n \in \Nset}$ be a sequence of monic OP 
in $L^{2}(\dif\mu)$. Assume that  the sequence of monic polynomials $(q_n)_{n \in \Nset}$ constructed via
\begin{equation} \label{eq:sop-construct}
\begin{split}
	q_{2 k}(z) &\defequals \sum_{j = 0}^{k} \mu_{k, j} \, p_{2 j}(z), \\
	q_{2 k + 1}(z) &\defequals p_{2 k + 1}(z),
\end{split}
\end{equation}
satisfies the skew-orthogonality conditions \eqref{eq:SOPdef1} and \eqref{eq:SOPdef2}. Then, the $(p_n)_{n \in \Nset}$ satisfy a three-term recurrence relation \eqref{eq:3term}.

Conversely, if the sequence $(p_n)_{n \in \Nset}$ of monic OP satisfies a three-term recurrence relation \eqref{eq:3term}, then, the  sequence of monic polynomials defined in \eqref{eq:sop-construct} forms a SOP system and the $\mu_{n,j}$'s in \eqref{eq:sop-construct} can be explicitly computed:
\begin{align}
	r_k &= 2 \p{h_{2 k + 1} - c_{2 k + 1} h_{2 k}} > 0, \label{eq:rk-construct} \\
	\mu_{k, j} &= \prod^{k - 1}_{l = j} \lambda_l, \quad
	\lambda_l = \frac{
		h_{2 l + 2} - c_{2 l + 2} h_{2 l + 1}
	}{
		h_{2 l + 1} - c_{2 l + 1} h_{2 l}
	} \quad
	\forall j < k \in \Nset. \label{eq:mu-construct}
\end{align}
\end{theorem}


\begin{proof}
Assuming that the sequence in \eqref{eq:sop-construct} satisfies the skew-orthogonality conditions \eqref{eq:SOPdef1} and \eqref{eq:SOPdef2},  we want to show the three-term recurrence relation for the OP. 
Therefore consider the following Fourier expansion, starting with the odd polynomials and expressing $zp_{2m+1}(z)$ in terms of the $p_{l}$'s,
\begin{equation}
zp_{2m+1}(z)=p_{2m+2}(z)+b_{2m+1}p_{2m+1}(z)+c_{2m+1}p_{2m}(z)+ f_{2m-1}(z),
\end{equation}
where we have to show that 
\begin{equation}
f_{2m-1}(z) \defequals \sum_{l=0}^{2m-1}a_{2m+1,l}\,p_{l}(z)
\end{equation}
is the null polynomial.
From $p_{2n+1}=q_{2n+1}$ we have for all $n, m$
\begin{equation*}
0 = \skp{p_{2n+1}}{p_{2m + 1}} = 2\Re\bk{
		\scp{z p_{2 n + 1}}{p_{2 m + 1}}
		- \scp{p_{2 n + 1}}{z p_{2 m + 1}}
	}.
\end{equation*}
Choosing in particular $n<m$, the first term vanishes due to orthogonality, $\scp{z p_{2n+1}}{p_{2m+1}}=0$.
From this we obtain $0 = \Re[\scp{p_{2n+1}}{z p_{2m+1}}] = a_{2m+1,2n+1}h_{2n+1}$, for $0 \leq n \leq m-1$. Consequently the odd coefficients vanish, due to $h_k \neq 0$ for all $k$.
Next, let us show that all even Fourier coefficients vanish too, $a_{2m+1,2l}=0$ for $0 \leq l \leq m-1$. For this, we use that $0 = \skp{q_{2l}}{q_{2m+1}}$ for $0 \leq l \leq m-1$, and the fact that $q_{2l}$ is a linear combination of even polynomials $p_{2j}$, see \eqref{eq:sop-construct}.
Then, by inspection starting with $l=0$, this implies that $0 = \skp{p_{2l}}{p_{2m+1}}$ for $0 \leq l \leq m-1$, and therefore for any linear combination of the $p_{2l}$ up to degree $2m-2$ as in $f_{2m-1}$.
In particular $0 = \skp{f_{2m-1}}{p_{2m+1}} = 2\Re[\scp{z f_{2m-1}}{p_{2m+1}} - \scp{f_{2m-1}}{z p_{2m+1}}]$. The first scalar product is trivially zero due to orthogonality, and the second term yields $0=\norm{f_{2m-1}}^2$, the desired result. 

Now we turn to the Fourier expansion of the even polynomials $zp_{2m}(z)$ given by 
\begin{equation}\label{expansioneven}
zp_{2m}(z)=p_{2m+1}(z)+b_{2m}p_{2m}(z)+c_{2m}p_{2m-1}(z)+ f_{2m-2}(z),
\end{equation}
with
\begin{equation}\label{eq:f2m-2def}
f_{2m-2}(z) \defequals \sum_{l=0}^{2m-2}a_{2m,l}\,p_{l}(z).
\end{equation}
By assumption $\mu_{m,j}=\lambda_{m-1}\mu_{m-1,j}$  and that $\lambda_{m-1}$  is independent of $j$, this implies that $q_{2m}(z) = p_{2m}(z) + \lambda_{m-1} q_{2(m-1)}(z)$. Using $0 = \skp{q_{2n}}{q_{2m}}$ for all $n, m$, we obtain $0 = \langle{q_{2n}},{p_{2m} + \lambda_{m-1} q_{2(m-1)}}\rangle_\text{s} = \skp{q_{2n}}{p_{2m}}$. Starting with $n=0$, by inspection we have $0 = \skp{p_{2n}}{p_{2m}}$. As before (choosing in particular $n<m$) this implies that all even Fourier coefficients $a_{2m,2n}$ vanish for $n=0,\ldots,m-1$.

Next we apply the skew-orthogonality $0 = \skp{q_{2n+1}}{q_{2m}}$ for $n<m$. In particular for $0 \leq l \leq m-2$, we have $0 = \skp{q_{2l+1}}{q_{2m}} = \langle{q_{2l+1}}{p_{2m} + \lambda_{m-1} q_{2(m-1)}}\rangle_\text{s} = \skp{p_{2l+1}}{p_{2m}}$.
Therefore, any linear combination of the $p_{2l+1}$ up to degree $2m-3$ (as in $f_{2m-2}$) will give zero in the skew-product. Applying the same argument as for the case of odd polynomials, we conclude that $0=\norm{f_{2m-2}}^2$. This concludes the proof of the first part of the theorem.

Conversely, let us assume that the OP satisfy a three-term recurrence relation, and define the $(q_n)_n$ by \eqref{eq:sop-construct}. For any polynomials $f, g$ with real coefficients, the evaluation of the skew-product \eqref{eq:skewproddef} can be reduced to the evaluation of the scalar product \eqref{eq:proddef}, $\skp{f}{g} = 2 \Re\bk{\scp{z f}{g} - \scp{f}{z g}}$. Thus, for the polynomials with odd indices we can write
\begin{equation*}
	\skp{q_{2 k + 1}}{q_{2 l + 1}}
	= \skp{p_{2 k + 1}}{p_{2 l + 1}}
	= 2 \Re\bk{
		\scp{z p_{2 k + 1}}{p_{2 l + 1}}
		- \scp{p_{2 k + 1}}{z p_{2 l + 1}}
	}.
\end{equation*}
To calculate the two scalar products we use the three-term recurrence relation and the orthogonality of the polynomials $p_n$. This leads to $\scp{z p_{2 k + 1}}{p_{2 l + 1}}= \scp{p_{2 k + 1}}{z p_{2 l + 1}}= b_{2 k + 1} h_{2 k + 1} \delta_{k, l}$. Hence, $\skp{q_{2 k + 1}}{q_{2 l + 1}} = 0$ for all $k$ and $l$. Similarly,  $\scp{p_{2 i}}{z p_{2 j}}=\scp{z p_{2 i}}{p_{2 j}}=b_{2 i}h_{2 i} \delta_{i, j}$, and after expanding the $q_{2n}$'s in terms of the $p_{2n}$'s, we obtain  $\skp{q_{2 k}}{q_{2 l}}=0$ $\forall k, l$.

For the combination of odd and even indices in \eqref{eq:SOPdef2} we need to evaluate 
\begin{align*}
	\scp{z p_{2 j}}{p_{2 l + 1}}
	&= 	h_{2 j + 1} \delta_{j, l} + c_{2 j} h_{2 j - 1} \delta_{j, l + 1} ,\quad \scp{p_{2 j}}{z p_{2 l + 1}}= h_{2 j} \delta_{j, l + 1} + c_{2 l + 1} h_{2 j} \delta_{j, l},
\end{align*}
which leads to
\begin{align*}
	\skp{q_{2 k}}{q_{2 l + 1}}
	&= \sum_{j = 0}^{k}
		2 \mu_{k, j} \Re[
			\scp{z p_{2 j}}{p_{2 l + 1}} - \scp{p_{2 j}}{z p_{2 l + 1}}
		] \\
	&= \sum_{j = 0}^{k}
		2 \mu_{k, j} \bk{
			\p{
			h_{2 j + 1} - c_{2 l + 1} h_{2 j}} \delta_{j, l}
			- \p{
			h_{2 j} - c_{2 j} h_{2 j - 1}} \delta_{j, l + 1}
		}.
\end{align*}
Depending on the values of $k$ and $l$ we need to distinguish between 3 cases:
\begin{enumerate}
\item Case $k < l$: \\
Here, we have $\delta_{j, l} = 0$ and $\delta_{j, l + 1} = 0$ for all $j \leq k$, and therefore $\skp{q_{2 k}}{q_{2 l + 1}} = 0$ as claimed.

\item Case $k = l$: \\
It holds $\delta_{j, l} = 1$ only for $j = l = k$, and $\delta_{j, l + 1} = 0$ for all $j \leq k$, therefore
\begin{equation}
\label{eq:rks}
	\skp{q_{2 k}}{q_{2 k + 1}}
	= 2 
	\p{
	h_{2 k + 1} - c_{2 k + 1} h_{2 k}}
	= r_k.
\end{equation}

\item Case $k > l$: \\
We have $\delta_{j, l} = 1$ only for $j = l < k$, and $\delta_{j, l + 1} = 1$ only for $j = l + 1 \leq k$. Then, we want that
\begin{equation*}
	\skp{q_{2 k}}{q_{2 l + 1}}
	= 2 \mu_{k, l} \p{
	h_{2 l + 1} - c_{2 l + 1} h_{2 l}}
	- 2 \mu_{k, l + 1} \p{
	h_{2 l + 2} - c_{2 l + 2} h_{2 l + 1}}
\end{equation*}
is zero.
This equation is fulfilled if and only if the $\mu_{k, l}$ satisfy
\begin{align}
\label{eq:murel}
	&\mu_{k, l + 1} \p{
	h_{2 l + 2} - c_{2 l + 2} h_{2 l + 1}}
	= \mu_{k, l} \p{
	h_{2 l + 1} - c_{2 l + 1} h_{2 l}} .
\end{align}
Thanks to our monic choice of the leading coefficient $\mu_{k,k}=1$, this recursion can be solved to obtain the explicit form as claimed
\begin{equation*}
	\mu_{k, j} = \prod^{k-1}_{l=j} \frac{
			h_{2 l + 2} - c_{2 l + 2} h_{2 l+1}
	}{
		h_{2 l + 1} - c_{2 l + 1} h_{2 l}
	} \quad
	\forall j < k \in \Nset.
\end{equation*}
In particular, note that the $\mu_{k,j}$ satisfy the assumption made in the first part of the theorem.
\end{enumerate}


\end{proof}

\begin{remark}
The recurrence coefficients $b_k$ (although non-zero in general) are not needed to express $\mu_{k, j}$ \eqref{eq:mu-construct} and $r_k$ \eqref{eq:rk-construct}.
Therefore, in the examples below, we will only give formulas for $c_k$ and $h_k$.

Also note that we can invert \eqref{eq:sop-construct} to get the following representation of the OP in terms of the SOP for all $k \in \Nset$:
\begin{equation}
\begin{split}
	p_{2k+2}(z) &= q_{2k+2}(z) - \lambda_k q_{2k}(z), \\
	p_{2k+1}(z) &= q_{2k+1}(z).
\end{split}
\end{equation}
Formula \eqref{eq:sop-construct} and its inverted form were observed before for the classical SOP 
on subsets of the real line based on Hermite, Laguerre and Jacobi polynomials \cite[Eqs.~(3.29), (3.28)]{AFNM}.
\end{remark}
For radially symmetric weight functions  $w(z)=w(\abs{z})$, the OP 
are given by  monomials. Due to $c_k=0\, (=b_k)$ for all $k\in\Nset$ in that case, the above formulas simplify considerably (this result is known, see~\cite[Eqs.~(31), (32)]{Jesper}):

\begin{corollary} \label{thm:sop_from_radial_weight}
If the weight function of the planar OP in \autoref{thm:sop_from_op} is radially symmetric, $w(z)=w(\abs{z})$, we have for the OP and their squared norms:
\begin{equation}
	p_n(z) = z^n, \quad
	h_n = 2 \pi \int_{0}^{\infty} r^{2 n + 1} w(r) \chi_{D_r}\dif r,
\end{equation}
where $\chi_{D_r}$ denotes the characteristic function on $D$ in radial direction. For the planar SOP we obtain:
\begin{align}
	q_{2 k}(z) &= z^{2k}+\sum_{j = 0}^{k-1}  z^{2 j}\prod^{k-1}_{l=j} \frac{
			h_{2 l + 2}}{
			h_{2 l + 1}} , \\
	q_{2 k + 1}(z) &= z^{2 k + 1} ,
	\label{eq:oddqrot}
\end{align}
with skew-norms
\begin{equation}
	r_k = 2 h_{2 k + 1}.
\end{equation}
\end{corollary}

Let us give some examples for planar SOP (and SOP over a weighted analytic Jordan curve) that are constructed using \autoref{thm:sop_from_op} or \autoref{thm:sop_from_radial_weight}. These immediately lead to examples of Pfaffian point processes \eqref{eq:eigenvalue_distribution} that are integrable, in the sense that the SOP and thus the corresponding (pre-) kernel is known explicitly. We will only state the SOP in what follows, and not the pre-kernel.


\begin{example}[Gegenbauer ensemble]
Consider the measure $\mu$ supported in the interior of the standard ellipse $\partial E$, with semi-axes $a > b > 0$, such that $\mu$ has density function $w$ with respect to planar Lebesgue measure:
\begin{equation}\label{eq:Eweight}
	w(z) = (1 + \alpha) \p{1 - h(z)}^\alpha,  \quad
	h(z) = \p{\frac{\Re(z)}{a}}^2 + \p{\frac{\Im(z)}{b}}^2,\; \alpha>-1.
\end{equation}
As shown in \cite[Thm.~3.1]{ANPV}, the Gegenbauer polynomials (also called ultraspherical) form an orthonormal basis on the Bergman space with respect to this weight function. The monic OP, recurrence coefficients and norms read
\begin{align}\label{eq:COP}
	p_n(z) &= \frac{n!}{(1+\alpha)_n}
		\p{\frac{c}{2}}^n
		C_n^{(1 + \alpha)}\p{\frac{z}{c}}, \quad c = \sqrt{a^2 - b^2} > 0, \\
		c_n&=\frac{n(n+1+2\alpha)}{(n+\alpha)(n+1+\alpha)}\p{\frac{c}{2}}^{2},
		\\
	h_n &= \pi a b \frac{1 + \alpha}{n + 1 + \alpha} \p{\frac{c}{2}}^{2n}
		\p{\frac{n!}{(1+\alpha)_n}}^2
		C_n^{(1 + \alpha)}\p{R}, \quad R = \frac{a^2 + b^2}{a^2 - b^2}.
\end{align}
From \autoref{thm:sop_from_op} we obtain
\begin{align}
	q_{2 k}(z) &= \sum_{j = 0}^{k} \mu_{k,j} \frac{(2 j)!\,\Gamma\p{1 + \alpha}}{\Gamma\p{2 j + 1 + \alpha}}
		\p{\frac{c}{2}}^{2 j}
		C_{2 j}^{(1 + \alpha)}\p{\frac{z}{c}} , \\
	q_{2 k + 1}(z) &= 
		\frac{(2 k + 1)!\,\Gamma\p{1 + \alpha}}{\Gamma\p{2 k + 2 + \alpha}}
		\p{\frac{c}{2}}^{2 k + 1}
		C_{2 k + 1}^{(1 + \alpha)}\p{\frac{z}{c}} ,
\end{align}
with coefficients
\begin{equation*}
	\mu_{k,_j}\! =\! 
	\prod_{l = j}^{k - 1} \!
	\frac{(2l+1)(2l+2)c^2}{4\p{2 l + 2 + \alpha}\p{2l+3 + \alpha} }
	\frac{
		(2 l + 2) C_{2 l + 2}^{(1 + \alpha)}\p{R}
		- (2 l + 3 + 2 \alpha) C_{2 l+1}^{(1 + \alpha)}\p{R}
	}{
		(2 l + 1) C_{2 l + 1}^{(1 + \alpha)}\p{R}
		- (2 l + 2 + 2 \alpha) C_{2 l}^{(1 + \alpha)}\p{R}
	},
\end{equation*}
and skew-norms
\begin{align}
	r_k &= 2\pi a b 
		\frac{2^{-2(2 k + 1)}((2k+1)!)^2c^{2(2 k + 1)}}{(2k+1)(1+\alpha)_{2k+1}(2+\alpha)_{2k+1}}	\bk{
			(2 k + 1) C_{2 k + 1}^{(1 + \alpha)}\p{R}
			- (2 k + 2 + 2 \alpha) C_{2 k}^{(1 + \alpha)}\p{R}
		}.
\end{align}
\end{example}

In the special case when $\alpha=0$ the weight function is constant and the Gegenbauer polynomials reduce to the Chebyshev polynomials of the second kind, $C_n^{(1)}(x)=U_n(x)$. Apart from the example with Gegenbauer polynomials, which are symmetric Jacobi polynomials, for non-symmetric Jacobi polynomials with parameters $ (\alpha + 1/2, \pm 1/2) $ the SOP can also be constructed. In particular they include the Chebyshev polynomials of the first up to fourth kind, see \cite{ANPV} and references therein.

\begin{example}[Truncated symplectic ensemble]
When $b \to a = 1$ in the above example, the eccentricity $\varepsilon$ of the ellipse takes its critical value $\varepsilon = 0$.
In this limit only the leading coefficient of the Gegenbauer polynomials survives and the monic OP become monomials. The weight function \eqref{eq:Eweight}, defined on the unit disk, reduces to
\begin{equation}\label{eq:TUweight}
	w(z) = (1+\alpha)\p{1 - \abs{z}^2}^\alpha, \quad
	\alpha >-1.
\end{equation}
The SOP coefficients and skew-norms take the form:
\begin{equation}
\left.\mu_{k,j}\right|_{a=b=1} = \prod_{l=j}^{k-1} \frac{2l+2}{2l+3+\alpha}, \quad
\left.r_k\right|_{a=b=1} = 2 \pi \Gamma(2+\alpha) \frac{\Gamma(2k+2)}{\Gamma(2k+3+\alpha)}.
\end{equation}
They also can be obtained directly from our \autoref{thm:sop_from_radial_weight}, since in this limit the weight is rotationally invariant.

This weight describes the truncated unitary ensemble \cite{ZS}: for integer values of $\alpha=N-M-1$ it appears in the eigenvalue distribution of Haar distributed unitary matrices of size $N$ truncated to the upper sub-block of size $M<N$.
The SOP 
can be used for the corresponding symplectic point process, see \cite{BL} for further details.

Also, when the weight function includes an extra charge insertion at the origin, that is, $w(z) = \abs{z}^{2c}(1 - \abs{z}^2)^\alpha, \alpha >-1$, $c>-1$, the $\mu_{j,k}$ and $r_k$ can be readily obtained.
\end{example}


\begin{example}[Mittag-Leffler ensemble]\label{ex:ML}
The weight function, that leads to a Bergman kernel in terms of the two-parametric {Mittag-Leffler} function \cite[Thm 1.6]{AKS}, is given by 
\begin{equation}\label{eq:MLweight}
	w(z) = \abs{z}^{2 c} \euler^{-\abs{z}^{2 \lambda}}, \quad \lambda > 0, \ c > -1,
\end{equation}
 including a charge insertion at the origin. 
When setting $\lambda=1$, it reduces to the induced Ginibre ensemble, see  \cite{Fischman} for the OP, and the SOP were determined in \cite{A05}. Setting also  $c=0$, it reduces to the standard Ginibre ensemble, see \cite{Mehta,Kanzieper}.
The squared norms $h_n$ of the  monic OP $p_n(z) = z^n$ with respect to the weight \eqref{eq:MLweight} can be easily obtained:
\begin{equation}\label{eq:MLOP}
	h_n = \frac{\pi}{\lambda}\Gamma\p{\frac{n + 1 + c}{\lambda}}.
\end{equation}
From \autoref{thm:sop_from_radial_weight} we can read off the monic planar SOP with respect to \eqref{eq:MLweight} and their skew-norms.
For $\lambda=1$ and $c=0$ this leads to the same pre-kernel as in the symplectic Ginibre ensemble \cite{Mehta}.
\end{example}

\begin{example}[Chebyshev on an elliptic contour]
Consider the measure $\mu$ supported on the standard ellipse $\partial E$, with semi-axes $a > b > 0$, such that $\mu$ has the density function: 
\begin{equation}
	w(z) = \sqrt{\abs{\frac{c + z}{c - z}}}, \quad
	c = \sqrt{a^2 - b^2}.
\end{equation}
Chebyshev polynomials of the third kind $V_n$ satisfy an orthogonality relation on this contour \cite{Mason} (here $\abs{\dif z}$ stands for the  arc length measure):
\begin{align}
	&\scp{p_n}{p_m} = \int_{\partial E} p_n(z) \cconj{p_m(z)} w(z) \abs{\dif z} = h_n \delta_{n, m} \quad \text{for} \\
	&p_n(z) = \p{\frac{c}{2}}^n V_n\p{\frac{z}{c}}, \quad
	h_n = \pi \frac{(a + b)^{2 n + 1} + (a - b)^{2 n + 1}}{2^{2 n}},\quad c_n = \frac{c^2}{ 4}.
\end{align}
Similar relations hold for Chebyshev polynomials of first, second and fourth kind for different weight functions \cite{Mason}.
From \autoref{thm:sop_from_op} we get the following Szeg\H{o} SOP 
\begin{equation}\begin{split}
		q_{2 k}(z) &= \sum_{j = 0}^{k} \mu_{k,j} \p{\frac{c}{2}}^{2 j} V_{2 j}\p{\frac{z}{c}}, \quad
		\mu_{k,j} = \frac{1}{2^{2 (k - j)}} \prod_{l=j}^{k-1} \frac{(a + b)^{4 l + 4} - (a - b)^{4 l + 4}}{(a + b)^{4 l + 2} - (a - b)^{4 l + 2}}, \\
		q_{2 k + 1}(z) &= \p{\frac{c}{2}}^{2 k + 1} V_{2 k + 1}\p{\frac{z}{c}},
\end{split}\end{equation}
and for the skew-norms we obtain
\begin{equation}
	r_k = \frac{\pi b}{2^{4 k}} \bk{(a + b)^{4 k + 2} - (a - b)^{4 k + 2}}.
\end{equation}

Note that in the limit $b \to a = 1$ this example simplifies to the uniform weight on the unit circle, for which the Pfaffian point process is already well known \cite{EMeckes}: it describes the eigenvalues of the circular quaternion ensemble (matrices of $\mathbb{S}p(2N)$ distributed according to Haar-measure).
As mentioned already, this point process is determinantal and can be analysed via OP. 
\end{example}


Further examples -- including rotationally invariant weights on $\Cset$ -- will be provided in Section~\ref{sec:PK-sop} and \autoref{appA}.
The case resulting from products of $M$ rectangular Ginibre matrices is deferred to the Appendix \ref{appA} as the resulting planar SOP are known \cite{Jesper}.

\section{Bergman-like kernel of skew-orthogonal polynomials}\label{sec:PK-sop}

In this section we will derive the Bergman-like kernel for planar skew-orthogonal Hermite and Laguerre polynomials in separate subsections. 
They are given by the limiting pre-kernel, the sum over orthonormal SOP, hence the terminology. In both cases the corresponding weights are given by a one parameter family of elliptic ensembles, see e.g. \eqref{eq:ell_weight} below for Hermite and Appendix \ref{appA}.
The proofs will be based on the rotationally invariant limit, that is when the 
parameter is chosen such that the 
underlying domain and weight function has rotational symmetry. In that case the corresponding limiting pre-kernels are known, see \cite{Kanzieper} and \cite{A05} respectively.
As a consequence, at the end of each subsection we will present  the universality of a one-parameter family of kernels  in such symplectic elliptic Ginibre ensembles, hinting at a much larger universality.

\subsection{Bergman-like kernel of skew-orthogonal Hermite polynomials}
In random matrix theory the planar Hermite polynomials appear in the solution of the elliptic complex Ginibre ensemble \cite{FKS98}, with the one-parameter complex potential
\begin{equation} \label{eq:ell_weight}
	Q_\tau(z) = \frac{1}{1 - \tau^2} \bk{\abs{z}^2 - \tau \Re(z^2)}
	=\frac{\Re(z)^2}{1+\tau}+\frac{\Im(z)^2}{1-\tau},
	 \quad 0 \leq \tau < 1, \quad w_\tau(z) = e^{-Q_\tau(z)}.
\end{equation}
The monic polynomials (here with their third recurrence coefficient $c_n$),
\begin{equation}\label{eq:monicH}
	p_n(z) = \left(\frac{\tau}{2}\right)^{\frac n2} H_n\p{\frac{z}{\sqrt{2 \tau}}}, \quad
	c_n = \tau n,
\end{equation}
satisfy
\begin{equation}\label{eq:HermiteOP}
	\int_{\Cset} p_n(z) \overline{p_m({z})} \, e^{-Q_\tau(z)} \dif A(z) = h_n \delta_{n,m}, \quad
	h_n = n!\,
\end{equation}
Here, $\dif A(z)$ denotes the area measure on the complex plane, divided by 
$\pi\sqrt{1 - \tau^2} $ to provide a probability measure, i.e. with $h_0=1$.
We refer to \autoref{ex-HSOP} for details, including a two-parameter complex normal distribution and the matching planar Hermite polynomials.
The corresponding Bergman kernel 
\cite{EM} is the standard one of Hermite polynomials on $\Rset$ \cite[18.18.28]{NIST}, after continuation in the arguments $\zeta,\eta\in\Cset$, $0\leq\tau<1$:
\begin{equation}\label{eq:Mehler}
K_\tau(\zeta,{\eta})=\sum_{n=0}^\infty\frac{1}{n!}\left(\frac{\tau}{2}\right)^{n}H_n\left(\frac{\zeta}{\sqrt{2\tau}}\right)H_n\left(\frac{\overline{\eta}}{\sqrt{2\tau}}\right) =
\frac{1}{\sqrt{1-\tau^2}} \exp\bk{\frac{1}{1-\tau^2} \p{\zeta\overline{\eta}-\frac{\tau}{2}(\zeta^2+\overline{\eta}^2)}}.
\end{equation}
This identity is also called Mehler-Hermite formula
or Poisson kernel.

The Hermite SOP and skew-norms with respect to the weight \eqref{eq:ell_weight} are known \cite{Kanzieper} and recollected in \autoref{ex-HSOP}, as following from \autoref{thm:sop_from_op}. We only give the resulting pre-kernel:
\begin{equation}\label{eq:preKHermite}
\begin{split}
\prekernel_{\tau,N}(z,u) &= 
	\sum_{k = 0}^{N - 1} \frac{1}{(2 k + 1)!!}\p{\frac{\tau}{2}}^{k+1/2}
	\sum_{l = 0}^k \frac{1}{(2 l)!!} \p{\frac{\tau}{2}}^l 
	\\
	&\quad \times
	\left[
		H_{2 k + 1}\left(\frac{z}{\sqrt{2\tau}}\right) H_{2 l}\left(\frac{u}{\sqrt{2\tau}}\right)
		- H_{2 l}\left(\frac{z}{\sqrt{2\tau}}\right) H_{2 k + 1}\left(\frac{u}{\sqrt{2\tau}}\right)
	\right].
\end{split}
\end{equation}
Here, the area measure is normalised by $\frac{1}{2\pi(1-\tau)\sqrt{1-\tau^2}}$, to achieve $r_0=1$ in analogy to above.
Both expressions \eqref{eq:Mehler} and \eqref{eq:preKHermite} are the sum over (skew-)orthonormal polynomials. 
We are thus led to consider the following limit which is the first main result of this section.
From now on we use the following notation $f_N(z)  \rightrightarrows f(z)$ to express that the sequence of functions $(f_N)_N$ converges to $f$ uniformly on any compact subset of $\Cset$ as $N \to \infty$.

\begin{theorem}[Bergman-like Hermite kernel] \label{thm:PKsopH}
Let $0<\tau<1$ and $z,u\in\Cset$, then we have that
	$\prekernel_{\tau,N}(z, u) \rightrightarrows 	S_{\tau}(z, u)$, 
given by
\begin{equation}\label{eq:PKsopH}
\begin{split}
S_{\tau}(z, u)&=	\sum_{k = 0}^{\infty} \frac{\p{\sfrac{\tau}{2}}^{k + \sfrac{1}{2}}}{(2 k + 1)!!}\!
	\sum_{l = 0}^k \frac{\p{\sfrac{\tau}{2}}^l}{(2 l)!!}\!
\left[
		H_{2 k + 1}\left(\frac{z}{\sqrt{2\tau}}\right) H_{2 l}\left(\frac{u}{\sqrt{2\tau}}\right)
		- H_{2 l}\left(\frac{z}{\sqrt{2\tau}}\right) H_{2 k + 1}\left(\frac{u}{\sqrt{2\tau}}\right)
	\right]
		\\
	&= \frac{\sqrt{\pi}}{\sqrt{2}(1+\tau)}
		\exp\left[\frac{1}{2(1 + \tau)} (z^2 + u^2)\right]
			\erf\p{\frac{1}{\sqrt{2(1 - \tau^2)}} (z - u)}.
\end{split}
\end{equation}
\end{theorem}
We use the standard notation for the error function, in the form  $\erf(x)=\frac{2x}{\sqrt{\pi}}\int_0^1\euler^{-x^2s^2}ds$, that can be continued to $x\in\Cset$. 
The result \eqref{eq:PKsopH} was already given in \cite[Eq. (18.6.53)]{KS}, without providing any details.
In analogy to \eqref{eq:Mehler} being the infinite sum over orthonormalised OP, we call the corresponding sum over orthonormal SOP \eqref{eq:PKsopH} Bergman-like Hermite kernel.

The proof of \autoref{thm:PKsopH} will be in two steps. First, we recall the rotationally invariant case $\tau=0$, see \autoref{Ginibreskernel}. In the second step we derive \eqref{eq:PKsopH} by using an integral representation for the Hermite polynomials together with \autoref{thm:hermite_origin_monomials}.

\begin{lemma} \label{thm:hermite_origin_monomials}
For $u, v \in \Cset$ define
\begin{equation}\label{gN}
	g_N(u, v) \defequals
		\sum_{k = 0}^{N - 1}
		\frac{u^{2 k + 1}}{(2 k + 1)!!}
		\sum_{l = 0}^{k}
		\frac{v^{2 l}}{(2 l)!!}.
\end{equation}
Then, as $N \to \infty$ we have $g_N(u,v) \rightrightarrows g(u,v)$, where the limiting function is given by
\begin{equation}\label{eq:preKpart1}
	g(u,v) = \frac{1}{2} \sqrt{\frac{\pi}{2}}
		\euler^{\frac{1}{2} \p{u^2 + v^2}}
		\bk{
			\erf\p{\frac{u - v}{\sqrt{2}}}
			+ \erf\p{\frac{u + v}{\sqrt{2}}}
		}.
\end{equation}
\end{lemma}

\begin{proof} Let $u,v \in B(0,r)$, where $B(a,r)$ denotes the disk of center $a$ and radius $r$. Then each summand in \eqref{gN}  is bounded by $re^{r^2}r^{2k}/(k!)$.
Hence, by the Weierstra{\ss} M-test, its sum is an analytic function of $u$ and $v$, being absolutely and uniformly convergent in each compact subset of the plane. Thus, the limiting function, as $N \to \infty$, is given by the power series
\begin{equation}
	g(u, v) =
		\sum_{k = 0}^{\infty}
		\frac{u^{2 k + 1}}{(2 k + 1)!!}
		\sum_{l = 0}^{k}
		\frac{v^{2 l}}{(2 l)!!}.
\end{equation}
Now, we take the derivative term-wise to obtain: $\partial_u g(u, v) = u g(u, v) + \cosh(u v)$, see \cite[Sect.~15.2 and App.~A.34]{Mehta} and  \cite[Sect.~5.2.1]{Kanzieper} for details.
A convenient initial value of this linear inhomogeneous ordinary differential equation is $u_0 = 0$, then $g(u_0, v) = 0$ for every $v \in \Cset$.
The solution of the initial value problem is given by (see \cite{H.Amann})
\begin{equation}\begin{split}
	g(u, v) &= \int_{0}^{u} \exp\bk{\int_{t}^{u} s \dif s} \cosh(v t) \dif t
	= \frac{1}{2} \euler^{\sfrac{u^2}{2}} \int_{0}^{u}
		\bk{\euler^{-\sfrac{t^2}{2} + v t} + \euler^{-\sfrac{t^2}{2} - v t}} \dif t \\
	&= \frac{1}{2} \sqrt{\frac{\pi}{2}} \euler^{\frac{1}{2} (u^2 + v^2)} \bk{
		\erf\p{\frac{u - v}{\sqrt{2}}} + \erf\p{\frac{u + v}{\sqrt{2}}}
	}.
\end{split}\end{equation}
In the last step we have used the definition of the error function provided below \eqref{eq:PKsopH}, as well as its anti-symmetry.
\end{proof}

\begin{remark}[Symplectic Ginibre pre-kernel]\label{Ginibreskernel}
Note that the right hand side of \eqref{eq:preKHermite} is a continuous function of the parameter $\tau$ in a neighbourhood of $\tau=0$.
\autoref{thm:hermite_origin_monomials} tells us that $S_{0}(u,v)$ 
is equal  to 
$g(u,v)-g(v,u)$:
\begin{equation}
	S_{0}(u,v)=
	g(u,v)-g(v,u)=\sum_{k = 0}^{\infty} \sum_{l = 0}^k
		\frac{
			u^{2 k + 1} v^{2 l} - u^{2 l} v^{2 k + 1}
		}{(2k + 1)!! \, (2 l)!!}
	= \sqrt{\frac{\pi}{2}} 
		\euler^{\frac{1}{2} \p{u^2 + v^2}}
		\erf\p{\frac{u - v}{\sqrt{2}}}.
		\label{eq:preKGin}
\end{equation}
This is a well-known limiting pre-kernel, first found by Mehta in \cite{Mehta} and later calculated in \cite{Kanzieper}. 
As can be seen from \eqref{eq:ell_weight},
the parameter $\tau$ controls the degree of Hermiticity of the elliptic ensemble. The case $\tau=0$ here corresponds to maximal non-Hermiticity, when
real and imaginary part share the same variance.
In Theorem \ref{thm:PKsopH} the limit $N\to\infty$ is taken at fixed $\tau$, which we call strong non-Hermiticity. The case when $\tau\to1$ at a rate depending on $N$ called weak non-Hermiticity will not be discussed further, and we refer to \cite{FKS98,Kanzieper} for details.
The kernel in \eqref{eq:preKGin} is the symplectic analogon of its Hermitian partner $\exp[u\overline{v}]$ in the holomorphic Fock-space \cite{VBarg} 
in $L^2(\euler^{-|z|^2}\dif A(z))$, as in Remark \ref{Rem2.6}.
\end{remark}

Let us turn to the second step, involving double-sums of Hermite polynomials.
We begin with the following Lemma.
\begin{lemma} \label{thm:hermite_origin_generalized}
Let $\zeta, \eta \in \Cset$, $\varphi, \phi \in [0, 1)$ and define
\begin{equation}
	f_N(\zeta, \eta) \defequals
	\sum_{k = 0}^{N - 1}
		\frac{\p{\sfrac{\varphi}{2}}^{k + \sfrac{1}{2}}}{(2 k + 1)!!}
		H_{2 k + 1}(\zeta)
		\sum_{l = 0}^{k}
		\frac{\p{\sfrac{\phi}{2}}^l}{(2 l)!!}
		H_{2 l}(\eta).
\end{equation}
Then, as $N \to \infty$ we have $f_N(\zeta, \eta) \rightrightarrows f_{\varphi,\phi}(\zeta,\eta)$, where the limiting function is given by
\begin{equation}\label{eq:flim}
\begin{split}
	f_{\varphi,\phi}(\zeta, \eta)
	&= \frac{\sqrt{\pi}}{2 \sqrt{2 (1 + \varphi) (1 + \phi)}}
		\exp\p{\frac{\varphi}{1 + \varphi} \zeta^2 + \frac{\phi}{1 + \phi} \eta^2} \\
	&\quad \times \bk{
			\erf\parentheses[\Big]{a(\varphi,\phi)\zeta-a(\phi,\varphi)\eta}
			+ \erf\parentheses[\Big]{a(\varphi,\phi)\zeta+a(\phi,\varphi)\eta}
		},
\end{split}
\end{equation}
with
\begin{equation}
a(\varphi,\phi)=\sqrt{\frac{\varphi(1+\phi)}{(1+\varphi)(1-\varphi\phi)}}.
\end{equation}
\end{lemma}

\begin{proof}
The Hermite polynomials have the integral representation (\cite[18.10.10]{NIST})
\begin{align}
	H_n(x)
	&= \frac{(- 2 \iunit)^n}{\sqrt{\pi}}
		\int_{-\infty}^{\infty} t^n \euler^{(\iunit t + x)^2} \dif t\,,
\end{align}
which we use to derive
\begin{equation}\label{estimateH}
	f_N(\zeta, \eta)
	= \frac{1}{\pi}
		\int_{\Rset^2}
		\euler^{(\iunit t + \zeta)^2 + (\iunit s + \eta)^2}
		\sum_{k = 0}^{N - 1}
		\sum_{l = 0}^{k}
		\frac{
			\p{- \iunit \sqrt{2 \varphi} t}^{2 k + 1}
			\p{- \iunit \sqrt{2 \phi} s}^{2 l}
		}{(2 k + 1)!! \, (2 l)!!}
		\dif s \dif t.
\end{equation}
A simple estimate of the integrand, as in \autoref{thm:hermite_origin_monomials}, leads to
\begin{equation}
	\abs{f_N(\zeta, \eta)} \leq \euler^{\Re(\zeta^2) + \Re(\eta^2)}\int_{\Rset^2}
		\abs{t} \euler^{-(1 - \varphi) t^2 - 2 \Im(\zeta) t}
		\euler^{-(1 - \phi) s^2 - 2 \Im(\eta) s}\dif s \dif t<\infty \quad \forall N.
\end{equation}
By Lebesgue's dominated convergence theorem we obtain
\begin{equation*}
\begin{split}
	\lim_{N \to \infty} f_N(\zeta, \eta)
	&= \frac{1}{2 \sqrt{2 \pi}}
		\euler^{\zeta^2 + \eta^2}
		\int_{\Rset^2}
		\euler^{-(1 + \varphi) t^2 - (1 + \phi) s^2 + 2 \iunit (t \zeta + s \eta)} \\
	&\quad \times
		\bk{
			\erf\p{- \iunit \sqrt{\varphi} t + \iunit \sqrt{\phi} s}
			+ \erf\p{- \iunit \sqrt{\varphi} t - \iunit \sqrt{\phi} s}
		} \dif s \dif t\,.
\end{split}
\end{equation*}
For the limit we used \autoref{thm:hermite_origin_monomials}, with $u = - \iunit \sqrt{2 \varphi} t$ and $v = - \iunit \sqrt{2 \phi} s$. Then, we can evaluate this integral via \eqref{eq:Gauss3} from \autoref{appC},  with $A = 1 + \varphi$, $B = 1 + \phi$, $C = - \iunit \sqrt{\varphi}$, $D = \pm \iunit \sqrt{\phi}$.
This leads to  the claimed limit \eqref{eq:flim}.
\end{proof}

Finally we can complete the proof of our main result.
\begin{proof}[Proof of \autoref{thm:PKsopH}]
As for \eqref{estimateH} we can obtain a simple estimate for the sum in \eqref{eq:PKsopH} with $\zeta, \eta \in B(0,r)$.
Thus, the sum converges absolutely and uniformly on $B(0,r)$, and therefore on each compact subset of $\Cset$. The absolute convergence allows us to rearrange the summands and, using \autoref{thm:hermite_origin_generalized} with $\varphi=\phi=\tau$, we obtain:
\begin{equation}
	S_{\tau}(z,u) = 
	f_{\tau,\tau}\left(\frac{z}{\sqrt{2\tau}},\frac{u}{\sqrt{2\tau}}\right)-f_{\tau,\tau}\left(\frac{u}{\sqrt{2\tau}},\frac{z}{\sqrt{2\tau}}\right).
\end{equation}
\end{proof}

\begin{remark} We note that the Bergman kernel $K_\tau$ in \eqref{eq:Mehler} of the analytic space in $L^2(e^{-Q_\tau} \dif A)$,
and the Bergman-like kernel $S_\tau$ in \eqref{eq:PKsopH}, both restricted to $\Rset \times \Rset$, are related by the following differential equation
\begin{equation}
	(1+\tau)\,e^{\frac{1}{2}(Q_\tau(x)+Q_\tau(y))} \partial_x \bk{e^{-\frac{1}{2}(Q_\tau(x)+Q_\tau(y))} S_\tau(x,y)} = K_\tau(x,y).
\end{equation}

\end{remark}

\subsubsection{Universality of the symplectic elliptic Ginibre kernel
}\label{sec:univH}

In this subsection we will prove the universality of all $k$-point correlation functions \eqref{eq:RkPf} in the symplectic elliptic Ginibre ensemble, in the large-$N$ limit at strong non-Hermiticity close to the origin. 
Throughout this article we have considered weight functions that are $N$-independent. Therefore, at large-$N$ the eigenvalues condense on a droplet with $N$-dependent radius, 
the support of the equilibrium measure. We refer to \cite{Ameur} for a discussion of the equilibrium problem for a general potential. 
In the case of our elliptic weight \eqref{eq:ell_weight}, we have that (\cite[Thm.~2.1]{B-GC} applied to the result of \cite{HJS}) yields
\begin{equation}\label{eq:ell-law}
	\corrfct_{N, 1}(z)\approx 
	\begin{cases}
		\frac{1}{2\pi(1-\tau^2)} & \text{if} \quad \p{\frac{\Re(z)}{1+\tau}}^2 + \p{\frac{\Im(z)}{1-\tau}}^2 \leq 2N, \\
		0 & \text{else}. \\
	\end{cases}
\end{equation}
The fact that in \autoref{thm:PKsopH} we take the limit of the kernel at arguments independent of $N$, 
implies that we consider the 
vicinity of the origin. In contrast, to investigate the neighbourhood of a bulk (or edge) point, we would have to 
centre the arguments around this point $z=\sqrt{2N} z_0$, with $\abs{z_0}$ of order unity, and rescale accordingly. 
Furthermore, because we keep $\tau$ fixed when $N\to\infty$, this implies that we consider strong non-Hermiticity.

Let us first quote the known result at maximal non-Hermiticity $\tau=0$ from \cite{Kanzieper}\footnote{Notice, that there the mean level spacing is rescaled to unity, compared to $2\pi$ in our case.}, the symplectic Ginibre kernel at the origin.
In that case, 
from \autoref{Ginibreskernel} 
we can read off  the matrix elements of the limiting kernel of $\kernel_N(z, u)$ in \eqref{eq:Kerneldef},
times the normalisation from the area measure
\begin{align}
\lim_{N\to\infty} \frac{1}{2\pi}\sqrt{w_{\tau=0}(z)w_{\tau=0}(u)}\prekernel_{0,N}(z,u) 
	&
	=\euler^{-\frac12 (|z|^2+|u|^2)} \frac{1}{2\pi}
	\sum_{k = 0}^{\infty} \sum_{l = 0}^k
		\frac{
			z^{2 k + 1} u^{2 l} - z^{2 l} u^{2 k + 1}
		}{(2k + 1)!! \, (2 l)!!}
		\nonumber\\
	&= \frac{1}{2\sqrt{2\pi}} \,
		\euler^{-\frac{1}{2} \p{|z|^2+|u|^2-z^2-u^2}}
		\erf\p{\frac{z - u}{\sqrt{2}}}.
		\label{eq:Gin-prekernel}
\end{align}
As a corollary from \autoref{thm:PKsopH}, we can now prove the following universality statement.
\begin{corollary}\label{Cor4.6}
The large-$N$ limit of the matrix elements $\sigma_N$ given in \eqref{eq:preKHermite} of the kernel $K_N$ \eqref{eq:Kerneldef} with respect to the weight function $w_\tau$ in \eqref{eq:ell_weight} 
are equivalent 
in the sense of Remark \ref{pfaffcocycle}
to \eqref{eq:Gin-prekernel} for general values of $0<\tau<1$ and thus universal.
\end{corollary}
\begin{proof}
In order to compare the 
limits 
of the pre-kernels \eqref{eq:preKHermite} and \eqref{eq:Gin-prekernel}
and their pre-factors, we have to map to the same equilibrium measure (also called macroscopic density) first, implying the same mean level spacing. This is called 
unfolding and we refer to \cite{Haake} for a standard reference that includes complex spectra, see also \cite{Ameur} for details about recentering and rescaling.
Because the symplectic Ginibre kernel in \eqref{eq:Gin-prekernel} is given with respect to the limiting density $\frac{1}{2\pi}$, in this case to unfold we simply have to rescale all arguments $z\to\sqrt{1-\tau^2} \, z$, in view of \eqref{eq:ell-law}. 
Therefore we consider the limit 
\begin{align}\label{limskhermit}
	\MoveEqLeft \lim_{N\to\infty} \frac{(1 - \tau^2)^{\frac32}}{2\pi(1-\tau)\sqrt{1-\tau^2}} \sqrt{w_{\tau}\p{\sqrt{1 - \tau^2} z} w_{\tau}\p{\sqrt{1 - \tau^2} u}} \, \prekernel_{\tau,N}\p{\sqrt{1 - \tau^2} z, \sqrt{1 - \tau^2} u} \nonumber \\
	&= \frac{1}{2\pi}(1+\tau)\, {\euler^{-\frac12 (|z|^2+|u|^2)+\frac{\tau}{4}(z^2+\cconj{z^2}+u^2+\cconj{u^2})} }
	\sum_{k = 0}^{\infty} \frac{\p{\sfrac{\tau}{2}}^{k + \sfrac{1}{2}}}{(2 k + 1)!!}
	\sum_{l = 0}^k \frac{\p{\sfrac{\tau}{2}}^l}{(2 l)!!}\nonumber\\
	&\qquad\times\left[
		H_{2 k + 1}\left(\frac{\sqrt{1-\tau^2}z}{\sqrt{2\tau}}\right) H_{2 l}\left(\frac{\sqrt{1-\tau^2}u}{\sqrt{2\tau}}\right)
		- H_{2 l}\left(\frac{\sqrt{1-\tau^2}z}{\sqrt{2\tau}}\right) H_{2 k + 1}\left(\frac{\sqrt{1-\tau^2}u}{\sqrt{2\tau}}\right)\right]\nonumber\\
	&=\euler^{\frac{\tau}{4}(\cconj{z^2}-z^2)}\euler^{\frac{\tau}{4}(\cconj{u^2}-u^2)}
		\frac{1}{2\sqrt{2\pi}} \,\euler^{-\frac{1}{2} \p{|z|^2+|u|^2-z^2-u^2}}
		\erf\p{\frac{z - u}{\sqrt{2}}}.
\end{align}
The additional factor $(1 - \tau^2)^{\frac32}$ in the first line originates from the rescaling of the measure $\dif A(z)$ and the factors $(\cconj{z} - z)$ in the Pfaffian representation of the correlation functions, eq.~\eqref{eq:RkPf}. In the first line we also multiply with the $\tau$-dependent normalisation of the area measure.
After inserting \eqref{eq:PKsopH} in the second line, we arrive at \eqref{eq:Gin-prekernel}, apart from two pre-factors.
These can be disregarded as they satisfy the condition under complex conjugation explained in \autoref{pfaffcocycle} to establish an equivalent kernel. Thus the universality of the kernel \eqref{eq:Gin-prekernel} holds for arbitrary fixed $\tau$, with $0\leq\tau<1$.
\end{proof}
After completing this work, the universality we have found at the origin at strong non-Hermiticity has been extended to the entire bulk (and edge) along the real line \cite{BE}. 
The universality of the elliptic Ginibre ensemble, a one-parameter family of Gaussian random matrices, strongly suggests a more general universality to hold, when allowing for a larger class of potentials $Q$ in the weight function \eqref{eq:ell_weight}. Despite our progress in Section \ref{sec:construct} in the construction of SOP for more general weight functions based on OP, this is beyond the scope of the current article.

\subsection{Bergman-like kernel of skew-orthogonal Laguerre polynomials}\label{sec:PK-Laguerre}
The generalised Laguerre polynomials, denoted by $L_n^{(\nu)}(x)$, are orthogonal on the interval $(0, \infty)$ with respect to the weight function $x^\nu \exp(-x)$, $\nu > -1$.
These polynomials also appear in the solution of the chiral elliptic complex Ginibre ensemble, which was introduced in \cite{James} to model Quantum Chromodynamics  with a baryon chemical potential $\mu$.
The weight function, defined on the complex plane, reads
\begin{equation} \label{eq:chEll_weight}
	w_\tau^{(\nu)}(z) = \abs{z}^{\nu}
		K_{\nu}\p{\frac{2}{1 - \tau^2} \abs{z}}
		\exp\left[\frac{2 \tau}{1 - \tau^2} \Re(z)\right], \quad
		0 \leq \tau < 1,
\end{equation}
where $K_\nu(z)$ is the modified Bessel function of the second kind.
We use the notation with the non-Hermiticity parameter $\tau$, instead of the chemical potential $\mu = \sqrt{1 - \tau}$ as in \cite{James}.
The monic polynomials (with their third recurrence coefficient $c_n$) are
\begin{equation}\label{eq:monicL}
	p_n(z) = (-1)^n n!\, \tau^n L_n^{(\nu)}\p{\frac{z}{\tau}}, \quad
	c_n = \tau^2 n (n + \nu),
\end{equation}
and they fulfil
\begin{equation}\label{eq:LaguerreOP}
	\int_{\Cset} p_n(z) \overline{p_m({z})} w_\tau^{(\nu)}(z) \dif A (z) = h_n \delta_{n,m}, \quad
	h_n = 
	n!\, \frac{\Gamma(n + \nu+1)}{\Gamma(\nu+1)},
\end{equation}
see \autoref{ex-LSOP} for details. Here, the area measure is divided by 
$\frac{\pi}{2} (1 - \tau^2)\Gamma(\nu+1)$, to achieve $h_0=1$.
The orthogonality on the complex plane was proven in \cite{Karp}, and independently in \cite{A05}.
The Poisson kernel for Laguerre polynomials is given by  \cite[18.18.27]{NIST} 
\begin{equation}\label{Bergmanforlaguerre}
 K_\tau(\zeta,\eta)=
 \sum_{n = 0}^{\infty} \frac{n!\, \Gamma(\nu+1)\,\tau^{2 n}}{\Gamma\p{n + \nu + 1}} L_n^{(\nu)}\left(\frac{\zeta}{\tau}\right) L_n^{(\nu)}\left(\frac{\overline{\eta}}{\tau}\right)
	= \frac{\Gamma(\nu+1)\exp\left[- \frac{\tau}{1 - \tau^2} (\zeta + \overline{\eta})\right]}{(1 - \tau^2)(\zeta \overline{\eta})^{\nu / 2}}
		I_{\nu}\p{\frac{2}{1 - \tau^2} \sqrt{\zeta \overline{\eta}}},
\end{equation}
where $I_\nu(z)$ denotes the modified Bessel function of the first kind. This identity corresponds to the Bergman kernel 
of the analytic space in $L^2(w_\tau^{(\nu)} \dif A)$.

The Laguerre SOP and skew-norms with respect to the weight \eqref{eq:chEll_weight} can be constructed from \autoref{thm:sop_from_op} (see \autoref{ex-LSOP}) and agree with \cite{A05}, with the resulting pre-kernel
\begin{equation}\label{eq:preKLaguerre}
\begin{split}
	\prekernel_{\tau,N}(z, u) &= 
	- \sum_{k = 0}^{N - 1} \frac{\sqrt{\pi}\Gamma(\nu+2)(2 k)!! \tau^{2 k + 1}}{2^{k+\nu+1} \Gamma\p{k + \frac{\nu}{2} + \frac{3}{2}}}
		\sum_{l = 0}^{k} \frac{(2 l - 1)!! \tau^{2 l}}{2^l \Gamma\p{l + \frac{\nu}{2} + 1}} \\
		&\qquad \times 
	\bk{
			L_{2 k + 1}^{(\nu)}\p{\frac{z}{\tau}} L_{2 l}^{(\nu)}\p{\frac{u}{\tau}}
			- L_{2 l}^{(\nu)}\p{\frac{z}{\tau}} L_{2 k + 1}^{(\nu)}\p{\frac{u}{\tau}}
		}.
\end{split}
\end{equation}
The normalising factor for the area measure is $\frac{1}{{\pi} (1 - \tau^2)^2\Gamma(\nu+2)}$, to have again $r_0=1$.
Expressions \eqref{Bergmanforlaguerre} and \eqref{eq:preKLaguerre} are the sum over (skew-)orthonormal 
Laguerre polynomials. 
In analogy to \eqref{Bergmanforlaguerre} we thus call the limiting sum  \eqref{eq:PKsopL} below Bergman-like Laguerre kernel.

The main result of this subsection is the following limit.
\begin{theorem}[Bergman-like Laguerre kernel] \label{thm:PKsopL}
Let $0<\tau<1$ and $z,u\in\Cset$, then for $N \to \infty$ we have
	$\prekernel_{\tau,N}(z,u) \rightrightarrows
	S_{\tau}(z,u)$, 
where 
\begin{equation} \label{eq:PKsopL}
\begin{split}
	S_\tau(z,u)&=
	-\sum_{k = 0}^{\infty} \frac{\sqrt{\pi}\Gamma(\nu+2)(2 k)!! \tau^{2 k + 1}}{2^{k+\nu+1} \Gamma\p{k + \frac{\nu}{2} + \frac{3}{2}}}
		\sum_{l = 0}^{k} \frac{(2 l - 1)!! \tau^{2 l}}{2^l \Gamma\p{l + \frac{\nu}{2} + 1}}
\\
&\qquad\qquad\qquad\qquad\qquad\qquad\qquad\times	
		\bk{
			L_{2 k + 1}^{(\nu)}\left(\frac{z}{\tau}\right) L_{2 l}^{(\nu)}\left(\frac{u}{\tau}\right)
			- L_{2 l}^{(\nu)}\left(\frac{z}{\tau}\right) L_{2 k + 1}^{(\nu)}\left(\frac{u}{\tau}\right)
		} \\
	&= \frac{\Gamma(\nu+2)\ \euler^{-\frac{\tau}{1 - \tau^2} (z + u)}}{(1 - \tau^2) (z u)^{\nu / 2}}
		\int_{0}^{\pi / 2}
			\sinh\p{\frac{1}{1 - \tau^2} (z - u) \cos(\alpha)}
			I_{\nu}\p{\frac{2}{1 - \tau^2} \sqrt{zu} \sin(\alpha)}
		\dif \alpha.
\end{split}
\end{equation}
\end{theorem}

The proof of the above theorem will proceed in two steps.
First, we treat the rotationally invariant case $\tau = 0$, then we apply the integral representation of Laguerre polynomials to establish a lemma analogous to \autoref{thm:hermite_origin_generalized}.

\begin{lemma} \label{thm:laguerre_origin_monomials}
For $u, v \in \Cset$ define
	\begin{equation}\label{thm:laguerre_origin_monomials1}
		g_N(u, v) \defequals
			\sum_{k = 0}^{N - 1} \frac{u^{2 k + 1}}{2^k \Gamma\p{k + \frac{\nu}{2} + \frac{3}{2}} (2 k + 1)!!}
			\sum_{l = 0}^{k} \frac{v^{2 l}}{2^l \Gamma\p{l + \frac{\nu}{2} + 1} (2 l)!!}.
	\end{equation}
	Then, as $N \to \infty$ we have $g_N(u,v) \rightrightarrows g(u,v)$, where the limiting function is given by
	\begin{equation}\begin{split}
	g(u,v)&= \frac{2^{\nu}}{\sqrt{\pi} (u v)^{\nu / 2}} \int_{0}^{\pi / 2} \bk{
			\euler^{(u + v) \cos(\alpha)} J_{\nu}\p{2 \sqrt{u v} \sin(\alpha)}
			- \euler^{-(u - v) \cos(\alpha)} I_{\nu}\p{2 \sqrt{u v} \sin(\alpha)}
		} \dif \alpha.
	\end{split}\end{equation}
\end{lemma}

\begin{proof} Let $u,v \in B(0,r)$, then the same upper bound as in \autoref{thm:hermite_origin_monomials} serves as a summable upper bound for \eqref{thm:laguerre_origin_monomials1}, and it only depends $r$. Hence, by the Weierstra{\ss} M-test, its sum is an analytic function of $u$ and $v$. The convergence is absolute and uniform in each compact subset of the plane. 
The limiting function -- as $ N \to \infty$ -- is given by the power series:
\begin{equation}
	g(u, v) \defequals
		\sum_{k = 0}^{\infty} \frac{u^{2 k + 1}}{2^k \Gamma\p{k + \frac{\nu}{2} + \frac{3}{2}} (2 k + 1)!!}
		\sum_{l = 0}^{k} \frac{v^{2 l}}{2^l \Gamma\p{l + \frac{\nu}{2} + 1} (2 l)!!}.
\end{equation}

Next we derive a differential equation for this limit.
Following the same steps as in \cite[App.~B.1]{A05},  
where a differential equation was derived for $g(u,v)-g(v,u)$, we can be brief. 
We obtain
\begin{equation}\begin{split}
	u^{-\nu} \partial_u u^{\nu + 1} \partial_u g(u, v)
	&= \sum_{k = 0}^{\infty} \frac{u^{2 k}}{2^{k - 1} \Gamma\p{k + \frac{\nu}{2} + \frac{1}{2}} (2 k - 1)!!}
		\sum_{l = 0}^{k} \frac{v^{2 l}}{2^l \Gamma\p{l + \frac{\nu}{2} + 1} (2 l)!!} \\
	&= \sum_{k = 0}^{\infty} \frac{u^{2 k + 2}}{2^{k} \Gamma\p{k + \frac{\nu}{2} + \frac{3}{2}} (2 k + 1)!!}
		\sum_{l = 0}^{k} \frac{v^{2 l}}{2^l \Gamma\p{l + \frac{\nu}{2} + 1} (2 l)!!} \\
	&\quad + \sum_{k = 0}^{\infty} \frac{(u v)^{2 k}}{2^{k - 1} \Gamma\p{k + \frac{\nu}{2} + \frac{1}{2}} (2 k - 1)!! 2^k \Gamma\p{k + \frac{\nu}{2} + 1} (2 k)!!}.
\end{split}\end{equation}
In the last step we have taken out the $l = k$ term of the inner sum, before shifting the index $k \to k + 1$.
The first sum that we are left with is equal to $u g(u, v)$. In the last sum we can simplify the denominator, by using the doubling formula for the gamma function and $(2 k - 1)!! (2 k)!! = (2 k)!$, leading to
\begin{equation}
 \frac{2^{\nu + 1}}{\sqrt{\pi}} \sum_{k = 0}^{\infty} \frac{(u v)^{2 k}}{\Gamma\p{2 k + \nu + 1} (2 k)!} 
	= \frac{2^{\nu}}{\sqrt{\pi} (u v)^{\nu / 2}} \p{I_{\nu}\p{2 \sqrt{u v}} + J_{\nu}\p{2 \sqrt{u v}}},
\end{equation}
where we used the series representations of the Bessel-J  and Bessel-I functions (\cite[10.2.2 and 10.25.2]{NIST})
\begin{equation}\label{def:Bessel's}
	J_{\nu}(z) = \p{\frac{z}{2}}^\nu \sum_{n = 0}^{\infty} \frac{(-1)^n \p{z / 2}^{2 n}}{n! \Gamma\p{n + \nu + 1}}, \quad
	I_{\nu}(z) =  \p{\frac{z}{2}}^\nu \sum_{n = 0}^{\infty} \frac{\p{z / 2}^{2 n}}{n! \Gamma\p{n + \nu + 1}}.
\end{equation}
In summary, $g(u, v)$ satisfies the following second order linear inhomogeneous differential equation
\begin{equation}
	\p{u \partial_u^2 + (\nu + 1) \partial_u - u} g(u, v)
	= \frac{2^{\nu}}{\sqrt{\pi} (u v)^{\nu / 2}} \left(I_{\nu}\p{2 \sqrt{u v}} + J_{\nu}\p{2 \sqrt{u v}}\right).
\end{equation}

To solve this equation, we may again use $u_0 = 0$ as an initial value, because $g(u_0, v) = 0$ for every $v \in \Cset$ (and we will see later that a second initial condition is not needed for the uniqueness of the solution).
We will solve this initial value problem in three steps: First, we find two linearly independent solutions $\gamma_{\text{homA}}(u)$ and $\gamma_{\text{homB}}(u)$ for the homogeneous equation, then we construct a solution $\gamma_{\text{inhom}}(u)$ for the inhomogeneous equation, finally we set $g(u, v) = a \gamma_{\text{homA}}(u) + b \gamma_{\text{homB}}(u) + \gamma_{\text{inhom}}(u)$ where $a, b \in \Cset$ are determined by the initial condition.

A simple computation shows that if $f(u)$ solves 
$u^2 f^{\prime \prime}(u) + u f^\prime(u) - (u^2 + \nu^2 / 4) f(u) = 0$, 
the (modified) Bessel-ODE, then $f(u) / u^{\nu / 2}$ solves the homogeneous ODE for $g(u, v)$.
Therefore we obtain
\begin{equation}
	\gamma_{\text{homA}}(u) = \frac{I_{\nu / 2}(u)}{u^{\nu / 2}}, \quad
	\gamma_{\text{homB}}(u) = \frac{K_{\nu / 2}(u)}{u^{\nu / 2}},
\end{equation}
where $K_\alpha$ is the modified Bessel function of the second kind.
As was shown in \cite[App.~B.2]{A05}, it holds
\begin{equation}
	\frac{I_{\nu}\p{2 \sqrt{u v}}}{2^{\nu} (u v)^{\nu / 2}}
	= \p{u \partial_u^2 + (\nu + 1) \partial_u - u} \p{- \frac{4}{\p{2 \sqrt{u v}}^{\nu}} \euler^{u + v}
			\int_{0}^{\infty} \int_{0}^{q} \euler^{-q^2 - p^2} J_{\nu}\p{2 q \sqrt{2 u}} J_{\nu}\p{2 p \sqrt{2 v}} \dif p \dif q
		},
\end{equation}
and similarly via the relation $J_{\nu}\p{2 \sqrt{u v}} = (-\iunit)^\nu I_{\nu}\p{2 \iunit \sqrt{u v}} = (-\iunit)^\nu I_{\nu}\p{2 \sqrt{u (-v)}}$ we get
\begin{equation}
	\frac{J_{\nu}\p{2 \sqrt{u v}}}{2^{\nu} (u v)^{\nu / 2}}
	= \p{u \partial_u^2 + (\nu + 1) \partial_u - u} \p{- \frac{4}{\p{2 \sqrt{u v}}^{\nu}} \euler^{u - v}
		\int_{0}^{\infty} \int_{0}^{q} \euler^{-q^2 - p^2} J_{\nu}\p{2 q \sqrt{2 u}} I_{\nu}\p{2 p \sqrt{2 v}} \dif p \dif q
	}.
\end{equation}
Hence, a solution for the inhomogeneous equation is
\begin{equation}
	\gamma_{\text{inhom}}(u) = - \frac{2^{\nu + 2}}{\sqrt{\pi} (u v)^{\nu / 2}} \euler^{u} \int_{0}^{\infty} \int_{0}^{q} \euler^{-q^2 - p^2} J_{\nu}\p{2 q \sqrt{2 u}} \bk{\euler^{v} J_{\nu}\p{2 p \sqrt{2 v}} + \euler^{-v} I_{\nu}\p{2 p \sqrt{2 v}}} \dif p \dif q.
\end{equation}
Using the series representations \eqref{def:Bessel's} of the Bessel functions $J_\nu$ and $I_\nu$, one can see that $\gamma_{\text{homA}}(u)$ and $\gamma_{\text{inhom}}(u)$ are continuous functions in a neighbourhood of $u=0$. Since $g(u,v)$ is continuous around $0$, with $g(0,v)=0$, we can set $b=0$. A simple calculation using \eqref{def:Bessel's} for $\gamma_{\text{homA}}(u)$ and \eqref{eq:Bessel Gamma integrals} together with \cite[10.32.2]{NIST} for $\gamma_{\text{inhom}}(u)$ gives
\begin{equation}
\begin{split}
	\lim_{u \to 0} \gamma_{\text{homA}}(u) = \frac{1}{2^{\nu / 2} \Gamma(\frac \nu2 + 1)}, \quad \lim_{u \to 0} \gamma_{\text{inhom}}(u)
	&= -\sqrt{\pi} \frac{2^{\nu / 2}}{\Gamma\p{\frac{\nu}{2} + 1}} \frac{I_{\nu / 2}(v)}{v^{\nu / 2}},
\end{split}
\end{equation}
and thus
\begin{equation}
	a = -2^{\nu / 2} \Gamma\p{\frac{\nu}{2} + 1},\quad  \gamma_{\text{inhom}}(0)
	= \sqrt{\pi} 2^{\nu} \frac{I_{\nu / 2}(v)}{v^{\nu / 2}}.
\end{equation}
By replacing the Bessel-$I$ functions with the integral representations \cite[10.22.52]{NIST} and \cite[10.43.24]{NIST}, we can match this term to one of the double-integrals in $\gamma_{\text{inhom}}(u)$, then
\begin{equation}\label{gBessel}
\begin{split}
	g(u, v) &= \frac{2^{\nu + 2}}{\sqrt{\pi} (u v)^{\nu / 2}} \left[
		\euler^{u - v} \int_{0}^{\infty} \int_{0}^{\infty} \euler^{-q^2 - p^2} J_{\nu}\p{2 q \sqrt{2 u}} I_{\nu}\p{2 p \sqrt{2 v}} \dif p \dif q \right. \\
		&\hphantom{= \frac{2^{\nu + 2}}{\sqrt{\pi} (u v)^{\nu / 2}} \quad}
			-\euler^{u + v} \int_{0}^{\infty} \int_{0}^{q} \euler^{-q^2 - p^2} J_{\nu}\p{2 q \sqrt{2 u}} J_{\nu}\p{2 p \sqrt{2 v}} \dif p \dif q \\
		&\hphantom{= \frac{2^{\nu + 2}}{\sqrt{\pi} (u v)^{\nu / 2}} \quad}
			\left. -\euler^{u - v} \int_{0}^{\infty} \int_{0}^{q} \euler^{-q^2 - p^2} J_{\nu}\p{2 q \sqrt{2 u}} I_{\nu}\p{2 p \sqrt{2 v}} \dif p \dif q \right] \\
	&= \frac{2^{\nu + 2}}{\sqrt{\pi} (u v)^{\nu / 2}} \left[
		\euler^{u - v} \int_{0}^{\infty} \int_{q}^{\infty} \euler^{-q^2 - p^2} J_{\nu}\p{2 q \sqrt{2 u}} I_{\nu}\p{2 p \sqrt{2 v}} \dif p \dif q \right. \\
		&\hphantom{= \frac{2^{\nu + 2}}{\sqrt{\pi} (u v)^{\nu / 2}} \quad}
			\left. -\euler^{u + v} \int_{0}^{\infty} \int_{0}^{q} \euler^{-q^2 - p^2} J_{\nu}\p{2 q \sqrt{2 u}} J_{\nu}\p{2 p \sqrt{2 v}} \dif p \dif q \right] \\
	&= \frac{2^{\nu + 2}}{\sqrt{\pi} (u v)^{\nu / 2}} \int_{0}^{\infty} \int_{0}^{q} \euler^{-q^2 - p^2} \bk{
		\euler^{u - v} J_{\nu}\p{2 p \sqrt{2 u}} I_{\nu}\p{2 q \sqrt{2 v}}
		- \euler^{u + v} J_{\nu}\p{2 q \sqrt{2 u}} J_{\nu}\p{2 p \sqrt{2 v}}
	}\! \dif p \dif q.
\end{split}
\end{equation}
In the last step we used Fubini's theorem, and then we switched the variable names $q$ and $p$.
Finally, the claimed formula for $g(u, v)$ follows from \eqref{eq:Bessel product double-integrals}.
\end{proof}

\begin{remark}[Chiral symplectic Ginibre pre-kernel]\label{prop:preKchGin}
Note that the right hand side of \eqref{eq:preKLaguerre} is a continuous function of the parameter $\tau$ in a neighbourhood of $\tau=0$.
\autoref{thm:laguerre_origin_monomials} tells us that $S_{0}(u,v)$ is 
proportional to $g(u,v)-g(v,u)$:
\begin{equation}\label{eq:chellS}
\begin{split}
	S_{0}(u,v)=
\frac{\sqrt{\pi}\Gamma(\nu+2)}{2^{\nu+1}}		\left(
	g(u,v)-g(v,u)
	\right)
	&= 
	\sum_{k = 0}^{\infty} \sum_{l = 0}^{k} \frac{\sqrt{\pi}\Gamma(\nu+2)\left(
			u^{2 k + 1} v^{2 l} - u^{2 l} v^{2 k + 1}\right)
		}{
			2^{k+\nu+1} \Gamma\p{k + \frac{\nu}{2} + \frac{3}{2}} (2 k + 1)!!
			2^l \Gamma\p{l + \frac{\nu}{2} + 1} (2 l)!!
		} \\
	&= \frac{\Gamma(\nu+2)}{(u v)^{\nu / 2}} \int_{0}^{\pi / 2} \sinh\parentheses[\Big]{(u - v) \cos(\alpha)}	I_{\nu}\parentheses[\Big]{2 \sqrt{u v} \sin(\alpha)} \dif \alpha.
\end{split}
\end{equation}
This is the known limiting kernel found in \cite[Appendix~B]{A05} and is expected to be universal. As in \autoref{Ginibreskernel}, $\tau=0$ considered here corresponds to  maximal non-Hermiticity of the underlying ensemble.

The kernel in \eqref{eq:chellS} is the symplectic analogon of its Hermitian partner $I_\nu(u\overline{v})$ (in terms of squared variables) of the generalised Fock-space \cite{FC}  in $L^2\left(|z|^{2\nu+2}K_\nu(2|z|^2)\dif A(z)\right)$.
\end{remark}

We can now turn to the double-sum of Laguerre polynomials.
\begin{lemma} \label{thm:laguerre_origin_generalized}
	Let $\zeta, \eta \in \Cset$, $\vartheta, \theta \in [0, 1)$ and define:
	\begin{equation}\label{partialsumlaguerre}
		f_N(\zeta, \eta) \defequals
		\sum_{k = 0}^{N - 1}
		\frac{(2k)!! \vartheta^{2 k + 1}}{2^k \Gamma\p{k + \frac{\nu}{2} + \frac{3}{2}}}
		L_{2 k + 1}^{(\nu)}(\zeta)
		\sum_{l = 0}^{k}
		\frac{(2 l - 1)!! \theta^{2 l}}{2^l \Gamma\p{l + \frac{\nu}{2} + 1}}
		L_{2 l}^{(\nu)}(\eta).
	\end{equation}
	Then, as $N \to \infty$ we have $f_N(\zeta, \eta) \rightrightarrows f_{\vartheta,\theta}(\zeta,\eta)$, where the limiting function is given by
	\begin{equation}\label{limlaguerretau}
	\begin{split}
		f_{\vartheta,\theta}(\zeta, \eta)
		&= \frac{2^\nu\sqrt{a(\vartheta)a(\theta)}}{\sqrt{\pi \vartheta \theta} (\vartheta \theta \zeta \eta)^{\nu / 2}} \euler^{- \vartheta a(\vartheta)\zeta - \theta a(\theta)\eta} 
	\left[ \int_{0}^{\lambda}
			\euler^{\p{-a(\vartheta)\zeta - a(\theta)\eta} \cos(t)}
			J_\nu\p{2 \sqrt{a(\vartheta)a(\theta)\zeta\eta} \, \sin(t)}
			\dif t \right. \\
		&\qquad\qquad\qquad\qquad\qquad\qquad\qquad\left.-\int_{0}^{\mu}
			\euler^{\p{a(\vartheta)\zeta - a(\theta)\eta} \cos(t)}
			J_\nu\p{2 \sqrt{a(\vartheta)a(\theta)\zeta\eta} \, \sin(t)}
			\dif t \right],
	\end{split}
	\end{equation}
	with
	\begin{equation}
		a(\theta) = \frac{\theta}{1-\theta^2}, \quad
		\lambda = 2 \arctan\p{\sqrt{\frac{(1 + \vartheta) (1 + \theta)}{(1 - \vartheta) (1 - \theta)}}}, \quad
		\mu = 2 \arctan\p{\sqrt{\frac{(1 - \vartheta) (1 + \theta)}{(1 + \vartheta) (1 - \theta)}}}.
	\end{equation}
\end{lemma}
\begin{proof}
We replace the Laguerre polynomials with the integral (compare \cite[18.10.9]{NIST})
\begin{equation}
	L_n^{(\nu)}(x) = \frac{2}{n!} x^{-\nu / 2} \euler^x
		\int_{0}^{\infty} t^{2 n + \nu + 1} \euler^{-t^2} J_\nu\p{2 \sqrt{x} t} \dif t.
\end{equation}
Then, we get
\begin{equation}
\begin{split}
	f_N(\zeta, \eta) &= 4 (\zeta \eta)^{-\nu / 2} \euler^{\zeta + \eta}
		\int_{0}^{\infty} \int_{0}^{\infty} (t s)^{1 + \nu} \euler^{-t^2 - s^2} J_\nu\p{2 \sqrt{\zeta} t} J_\nu\p{2 \sqrt{\eta} s} \\
	&\quad \times \sum_{k = 0}^{N - 1} \sum_{l = 0}^{k} \frac{(\vartheta t^2)^{2 k + 1}}{2^k \Gamma\p{k + \frac{\nu}{2} + \frac{3}{2}} (2 k + 1)!!} \frac{(\theta s^2)^{2 l}}{2^l \Gamma\p{l + \frac{\nu}{2} + 1} (2 l)!!}
		\dif t \dif s.
\end{split}
\end{equation}
This integral is uniformly bounded\footnote{We estimate the Bessel-$J$ functions with \cite[10.14.4]{NIST} and the double-sum with $\euler^{\vartheta t^2 + \theta s^2}$.} for $\zeta,\eta \in B(0,r)$ by
\begin{equation}
	\abs{f_N(\zeta, \eta)} \leq r^{\nu} \int_{\Rset_{+}^2}(t s)^{1 + 2 \nu} \exp\left[
		-(1 - \vartheta) t^2 - (1 - \theta) s^2 + 2rt + 2rs\right] \dif t \dif s < \infty, \quad \forall N.
\end{equation}
Using Lebesgue's dominated convergence theorem, we replace the limit of the double-sum with the double-integral from \eqref{gBessel} and evaluate the $t$- and $s$-integrals with \eqref{eq:Bessel product integrals}:
\begin{equation}\begin{split}
	f_{\vartheta,\theta}(\zeta, \eta) &= \frac{2^{\nu + 4}}{\sqrt{\pi}} \p{\vartheta \theta \zeta \eta}^{-\nu / 2} \euler^{\zeta + \eta} \int_{0}^{\infty} \int_{0}^{q} \euler^{-q^2 - p^2} \\
	&\quad \times \left[
		\int_{0}^{\infty} t \euler^{-(1 - \vartheta) t^2} J_\nu\p{2 \sqrt{\zeta} t} J_\nu\p{2 \sqrt{2 \vartheta} p t} \dif t
		\int_{0}^{\infty} s \euler^{-(1 + \theta) s^2} J_\nu\p{2 \sqrt{\eta} s} I_\nu\p{2 \sqrt{2 \theta} q s} \dif s \right. \\
	&\qquad \left.
		-\int_{0}^{\infty} t \euler^{-(1 - \vartheta) t^2} J_\nu\p{2 \sqrt{\zeta} t} J_\nu\p{2 \sqrt{2 \vartheta} q t} \dif t
		\int_{0}^{\infty} s \euler^{-(1 - \theta) s^2} J_\nu\p{2 \sqrt{\eta} s} J_\nu\p{2 \sqrt{2 \theta} p s} \dif s \right] \\
	&= \frac{2^{\nu + 2}}{\sqrt{\pi}} \p{\vartheta \theta \zeta \eta}^{-\nu / 2} \euler^{\zeta + \eta} \int_{0}^{\infty} \int_{0}^{q} \euler^{-q^2 - p^2} \\
	&\quad \times \left[ \frac{1}{(1 - \vartheta) (1 + \theta)}
		\exp\p{-\frac{\zeta + 2 \vartheta p^2}{1 - \vartheta} - \frac{\eta - 2 \theta q^2}{1 + \theta}}
		I_\nu\p{\frac{2 \sqrt{2 \vartheta \zeta} p}{1 - \vartheta}}
		J_\nu\p{\frac{2 \sqrt{2 \theta \eta} q}{1 + \theta}} \right. \\
	&\qquad \left. -\frac{1}{(1 - \vartheta) (1 - \theta)}
		\exp\p{-\frac{\zeta + 2 \vartheta q^2}{1 - \vartheta} - \frac{\eta + 2 \theta p^2}{1 - \theta}}
		I_\nu\p{\frac{2 \sqrt{2 \vartheta \zeta} q}{1 - \vartheta}}
		I_\nu\p{\frac{2 \sqrt{2 \theta \eta} p}{1 - \theta}} \right] \dif p \dif q.
\end{split}\end{equation}
Applying the substitution $p = q r$, $r \in [0, 1]$, we can calculate the $q$-integral with \eqref{eq:Bessel product integrals}.
In the last step we bring the two resulting one-dimensional integrals into the form in \eqref{limlaguerretau}. This is achieved by following changes of variables (for the first and second integral respectively)
\begin{equation*}
	r \to \sqrt{\frac{(1 - \vartheta) (1 - \theta)}{(1 + \vartheta) (1 + \theta)}} \, \tan(\alpha / 2), \quad
	r \to \sqrt{\frac{(1 + \vartheta) (1 - \theta)}{(1 - \vartheta) (1 + \theta)}} \, \tan(\alpha / 2).
\end{equation*}
\end{proof}
We can now complete the proof of our main result.
\begin{proof}[Proof of \autoref{thm:PKsopL}] \autoref{thm:laguerre_origin_monomials} together with \autoref{thm:laguerre_origin_generalized} shows that the sum in \eqref{partialsumlaguerre} converges absolutely and uniformly in each ball of center $0$ and radius $r$, and therefore in each compact subset of $\Cset$. Due to the absolute convergence of the series we can rearrange it and, in particular when $\vartheta = \theta = \tau$, we obtain 
\begin{equation}
S_\tau(z,u)=
\frac{\sqrt{\pi}\Gamma(\nu+2)}{2^{\nu+1}}	
\left(
f_{\tau,\tau}\left(\frac{z}{\tau},\frac{u}{\tau}\right)-f_{\tau,\tau}\left(\frac{u}{\tau},\frac{z}{\tau}\right)
\right).
\end{equation}
\end{proof}

\subsubsection{Universality of the symplectic chiral elliptic Ginibre kernel}\label{sec:univL}

In this subsection we will prove the universality of all $k$-point correlation functions \eqref{eq:RkPf} in the symplectic chiral elliptic Ginibre ensemble, in the large-$N$ limit at strong non-Hermiticity close to the origin.
The comments from Subsection \ref{sec:univH} about scaling and $N$-dependence of the weight apply here as well.
To derive the macroscopic density and the droplet for our weight \eqref{eq:chEll_weight}, we exploit the Bessel asymptotic $K_\nu(z) \sim \sqrt{\pi / (2 z)} \euler^{-z}$ to construct the limiting potential.
Then, from \cite[Thm.~2.1]{B-GC} together with \cite[Thm.~1]{ABK} we obtain
\begin{equation} \label{eq:chEll-law}
	\corrfct_{N, 1}(z) \approx \begin{cases}
		\frac{1}{4 \pi (1 - \tau^2) \abs{z}} & \text{if} \quad \p{\frac{\Re(z) - 4 \tau}{1 + \tau^2}}^2 + \p{\frac{\Im(z)}{1 - \tau^2}}^2 \leq 4 N, \\
		0 & \text{else}.
	\end{cases}
\end{equation}
For the behaviour at the origin let us first quote the known result at maximal non-Hermiticity $\tau = 0$ from \cite{A05}.
In that case we can read off from \autoref{prop:preKchGin} the matrix elements of the limiting kernel of $\kernel_N(z, u)$ from \eqref{eq:Kerneldef}, 
times the normalisation from the area measure
\begin{equation}
\begin{split}
	\MoveEqLeft \lim_{N \to \infty} 
\frac{1}{\pi\Gamma(\nu+2)}	
	\sqrt{w_{\tau=0}^{(\nu)}\p{z} w_{\tau=0}^{(\nu)}\p{u}} \, \prekernel_{0,N}\p{z, u}\\
	&= \left(\frac{|zu|}{zu}\right)^{\nu / 2} \frac{1}{\pi} \sqrt{K_\nu\p{2 \abs{z}} K_\nu\p{2 \abs{u}}} \int_{0}^{\pi / 2} \sinh\parentheses[\Big]{(z - u) \cos(\alpha)} I_\nu\parentheses[\Big]{2 \sqrt{z u} \sin(\alpha)} \dif \alpha.
	\label{eq:chGin-prekernel}
\end{split}
\end{equation}
We call this end result the limiting symplectic chiral Ginibre kernel at the origin, after removing the first factor due to \autoref{pfaffcocycle}.
We can now prove the following universality statement.
\begin{corollary}
The large-$N$ limit of the matrix elements 
$\sigma_N$ given in \eqref{eq:preKLaguerre} of the kernel $K_N$ \eqref{eq:Kerneldef} with respect to weight function $w_\tau^{(\nu)}$ \eqref{eq:chEll_weight} 
are equivalent to \eqref{eq:chGin-prekernel} 
in the sense of Remark \ref{pfaffcocycle} 
for general values $0< \tau <1$ and thus universal.
\end{corollary}

\begin{proof}
As explained already in the proof of Corollary \ref{Cor4.6} we have to unfold. In this case
we have to rescale all arguments $z \to (1 - \tau^2) z$. Therefore, we take the limit
\begin{equation}
\begin{split}
&	\lim_{N \to \infty} 
\frac{(1 - \tau^2)^3 }{\pi \Gamma(\nu + 2)(1 - \tau^2)^2}
\sqrt{w_{\tau}^{(\nu)}\parentheses[\Big]{(1 - \tau^2) z} w_{\tau}^{(\nu)}\parentheses[\Big]{(1 - \tau^2) u}} \, \prekernel_{\tau,N}\parentheses[\Big]{(1 - \tau^2) z, (1 - \tau^2) u} \\
	&= \euler^{-\iunit \tau \Im(z + u)} 
\left(\frac{|zu|}{zu}\right)^{\nu / 2} 	
	\frac{1}{\pi} \sqrt{K_\nu\p{2 \abs{z}} K_\nu\p{2 \abs{u}}} \,
		\int_{0}^{\pi / 2} \sinh\parentheses[\Big]{(z - u) \cos(\alpha)} I_\nu\parentheses[\Big]{2 \sqrt{z u} \sin(\alpha)} \dif \alpha.
\end{split}
\end{equation}
The pre-factor $(1 - \tau^2)^3$ 
originates from the rescaling of the 
arguments  
and the factors $(\cconj{z} - z)$ in \eqref{eq:RkPf}, times the normalisation from the area measure. 
Furthermore, we have multiplied with the $\tau$-dependent factor from the area measure. 
After inserting \eqref{eq:PKsopL} in the second line we arrive at  \eqref{eq:chGin-prekernel}, apart from the two pre-factors which 
lead to an 
equivalent kernel, cf. \autoref{pfaffcocycle}. Thus the universality of the kernel \eqref{eq:chGin-prekernel} holds.
\end{proof}
Once again we expect the universality found for the chiral elliptic Ginibre ensemble to hold for a more general class of weight functions, that share the same singularity
of the weight \eqref{eq:chEll_weight} at the origin.

\section{Christoffel perturbation for skew-orthogonal polynomials}\label{sec:CP}

In this section we will relate the SOP $q_n$ with respect to the weight function $w(z)$ to those $q_n^{(1)}$ skew-orthogonal with respect to the weight function $w^{(1)}(z)=|z-m|^2w(z)$. For OP on subsets of the real line such a relation between OP with respect to weights $w(x)$ and $(x-m)w(x)$ (or in fact $P(x)w(x)$ for a polynomial $P(x)$) is well known under the name of Christoffel perturbation, and determinantal formulas exist, compare \cite{BDS}. Such a transformation, including multiplication of the measure by a rational function, is closely related to the Darboux transformation of integrable systems. In the complex plane we consider quadratic factors $|z-m|^2$, in order to preserve the non-negativity of the resulting weight\footnote{In applications in physics, e.g. in field theory with chemical potential, linear factors $(z-m)$ leading to signed measures also play an important role, see \cite{James}.}. For planar OP (and also weighted Szeg\H{o} polynomials) such a Christoffel perturbation has already been studied for $w^{(M)}(z)=\prod_{l=1}^M|z-m_l|^2w(z)$ in \cite{AV}, from which we borrow the notation. There, determinantal formulas similar to those in \cite{BDS} have been derived for arbitrary $M$. 
For SOP, no such formulas were know. 
Only the polynomial kernel $\prekernel_n^{(M)}$
of $w^{(M)}$ was given in terms of the Pfaffian determinant of the polynomial kernel $\prekernel_n$ and odd SOP $q_n$ of $w$, see \cite{ABa}. We will use these expressions to provide an explicit representation of the perturbed SOPs $q_n^{(1)}$ in the following theorem. 

\begin{theorem}\label{thm:CP_SOP}
Let $(q_n)_{n\in\Nset}$ be the family of monic SOP with respect to the weight function $w(z)$, with norms $r_n$ and pre-kernel $\prekernel_n(z,u)$. Then, the following expressions hold for the monic SOP $q_n^{(1)}(z)$, their norms $r_n^{(1)}$ and kernel  $\prekernel_n^{(1)}(z,u)$ with respect to the perturbed weight $w^{(1)}(z)=|z-m|^2w(z)$, where we assume $m\in\Rset$:
\begin{equation}\label{SOP1}
\begin{split}
	q_{2 n}^{(1)}(z) &= \frac{r_n\prekernel_{n + 1}(m, z)}{ (m - z)q_{2 n}(m)} , \\
	q_{2 n + 1}^{(1)}(z) &= \frac{
			q_{2 n + 2}(m) q_{2 n}(z)
			- q_{2 n}(m) q_{2n + 2}(z)}{(m - z)q_{2 n}(m)} + d_n q_{2 n}^{(1)}(z), \\
	r_n^{(1)} &= r_n \frac{q_{2n + 2}(m)}{q_{2 n}(m)},
\end{split}
\end{equation}
where $d_n \in \Rset$ is an arbitrary constant. Furthermore, it holds
\begin{equation}\label{prek1}
	\prekernel_{n + 1}^{(1)}(z, u) = \frac{\prekernel_{n + 1}(z, u) q_{2n + 2}(m) - \prekernel_{n + 1}(z, m) q_{2n + 2}(u) + \prekernel_{n + 1}(u, m) q_{2n + 2}(z)}{(m - z) (m - u) q_{2n + 2}(m)}.
\end{equation}
\end{theorem}
Notice that for $z=m$, in eqs.~\eqref{SOP1} and \eqref{prek1} both numerator and denominator vanish, leading to a finite expression after applying the rule of l'H{\^ o}pital.
For comparison we state here the corresponding result for OP from \cite{AV}, where the same statement about $z=m$ applies.

\begin{remark}[Perturbed OP] \label{cor:OP1}
Let us assume that $\mu$ has density function $w$ on some domain $\subseteq \Cset$, $m \in \Cset$. Then it follows from \cite{AV}, that one can express the sequence $(p_n^{(1)})_n$ of OP in $L^2(|\cdot - m|^2w(\cdot))$ in terms of the sequence $(p_n)_n$ of OP in $L^2(w(\cdot))$ as:
\begin{equation}\label{pertOP}
	p_n^{(1)}(z) = \frac{
			K_{n + 1}(z, {m}) p_{n + 1}(m)
			- K_{n + 1}(m, {m}) p_{n + 1}(z)}{(m - z) K_{n + 1}(m, {m})},
\end{equation}
where $K_{n+1}(z,u)$ is the polynomial kernel constituted by the partial sum -- up to $n$ -- of  the orthonormal polynomials $p_k/\sqrt{h_k}$.  For the norms $h_n^{(1)}$  and polynomial kernel $K_{n}^{(1)}$, we have
\begin{equation}\label{pertOP1}
\begin{split}
	h_n^{(1)} &= h_{n + 1}\frac{K_{n + 2}(m, {m})}{K_{n + 1}(m, {m})},\\
	K_{n}^{(1)}(z, u) &= \frac{K_{n + 1}(m, {m}) K_{n + 1}(z, u) - K_{n + 1}(z,{m}) K_{n + 1}(m, u)}{(m - z) (\overline{m} - \overline{u}) K_{n + 1}(m, {m})}.
\end{split}
\end{equation}
\end{remark}
The proof of this remark can be found as a special case in \cite[Section~3]{AV}, where both $p_n^{(M)}(z)$ and $K_{n}^{(M)}(z, u)$ are given in terms of a ratio of two determinants of sizes $M+1$ and $M$, respectively.
Analogously we can express $q_n^{(M)}$ and $\prekernel_n^{(M)}$ as a ratio of Pfaffians following \cite{ABa}, but for the proof of the above theorem we will only consider the simplest case $M = 1$.

We note that the simple relationship between the odd SOP and odd OP found in \autoref{thm:sop_from_op} breaks down for $q_{2 k + 1}^{(1)}(z) $ and $p_{2k+1}^{(1)}(z) $, even if the initial OP $p_n(z)$ were to satisfy a three-term recurrence relation.  
\begin{corollary}
Let $\mu$ be a measure that cannot be made rotationally symmetric under an affine transformation.
The sequence of OP 
$(p_n^{(1)})_n$ in $L^2(|\cdot - m|^2\dif\mu)$ does not satisfy a three-term recurrence relation. 
\end{corollary}
For a particular case of $p_n^{(1)}(z)$ given in terms of Gegenbauer polynomials $p_n(z)$, this was indeed shown in \cite[Section~5]{ANPV}.

Let us present the proof now for \autoref{thm:CP_SOP}.
\begin{proof}
We start with eq.~\eqref{prek1}. In \cite[Eq.~(2.14)]{ABa} the density $\corrfct_{N, 1}^{(1)}(z)$, defined as in \eqref{eq:Rkdef}, was expressed in terms of a ratio of Pfaffian determinants for an arbitrary product of $M$ characteristic polynomials, which for $M=1$ reads:
\begin{align}
	\corrfct_{N, 1}^{(1)}(z) &= \p{\cconj{z} - z} w^{(1)}(z) \prekernel_N^{(1)}(z, \cconj{z}) \nonumber\\
	&= \p{\cconj{z} - z} w^{(1)}(z)
		\frac{\prekernel_{n + 1}(z, \cconj{z}) q_{2n + 2}(m) - \prekernel_{n + 1}(z, m) q_{2n + 2}(\cconj{z}) + \prekernel_{n + 1}(\cconj{z}, m) q_{2n + 2}(z)}{(m-z)(m-\cconj{z})q_{2n + 2}(m)}.
\end{align}
In the first line we started with eq.~\eqref{eq:R1}  that holds for arbitrary weight functions. For the second step we used the result  in \cite[Eq. (2.14)]{ABa}, to establish the claim in eq.~\eqref{prek1}. From this equation we will recover the skew-norms $r_n^{(1)}$ and the even and odd SOP $q_n^{(1)}$, by taking appropriate limits.

First, we determine the skew-norms. 
From the definition \eqref{eq:prek-def}, together with the fact that the SOP are monic $q_n(z)\sim z^n$, 
with $\sim$ meaning that  $\lim_{\abs{z}\to\infty}q_n(z)z^{-n}=1$, 
we can read off the asymptotic for the pre-kernel with both arguments being large, $\abs{z},\abs{u} \gg 1$: 
\begin{equation}\label{prekasymp2}
	\prekernel_{n+1}(z,u) \sim \frac{z-u}{r_n}(zu)^{2n}.
\end{equation}
Similarly, we obtain the asymptotic for a single argument being large, $\abs{z} \gg 1$:
\begin{equation}\label{prekasymp1}
	\prekernel_{n+1}(z,u) \sim \frac{1}{r_n}z^{2n+1}q_{2n}(u).
\end{equation}
Next, we insert this into \eqref{prek1} to determine the leading order expansion of $\prekernel_{n+1}^{(1)}(z,u)$ for both arguments being large: 
\begin{align}
	\prekernel_{n+1}^{(1)}(z,u) &\sim \frac{(z-u)r_n^{-1}(zu)^{2n}q_{2n + 2}(m) - r_n^{-1}z^{2n+1} q_{2n}(m) u^{2n + 2} + r_n^{-1}u^{2n+1} q_{2n}(m) z^{2n+2}}{zu\,q_{2n + 2}(m)}\nonumber\\
	&\sim \frac{(z-u)(zu)^{2n}q_{2n}(m)}{r_nq_{2n+2}(m)}=\frac{z-u}{r_n^{(1)}}(zu)^{2n}.
\end{align}
In the second step we have only kept the leading order, and in the last step we used that \eqref{prekasymp2} also holds for $\prekernel_{n+1}^{(1)}$. Comparing the last two expressions, this leads to $r_n^{(1)}$ as in the last equation of \eqref{SOP1}.

Likewise, we can use \eqref{prekasymp1} together with \eqref{prek1} to read off 
the even SOP:
\begin{align}
	q_{2n}^{(1)}(z) &= \lim_{\abs{u}\to\infty} \frac{-r_n^{(1)}\prekernel_{n+1}^{(1)}(z,u)}{u^{2n+1}}
	= \frac{r_n}{(z-m)q_{2n}(m)} \prekernel_{n+1}(z,m),
\end{align} 
where we have inserted $r_n^{(1)}$ as well, to express the right hand side in terms of unperturbed quantities only. 
This agrees with the first equation in \eqref{SOP1}, upon using the anti-symmetry of the pre-kernel.

For the odd polynomials $q_{2n+1}^{(1)}(z)$ we have to go beyond the leading order and use that they are only determined up to an arbitrary constant times the even SOP of one degree less, $q_{2n}^{(1)}(z)$. For that purpose we label the next-to-leading order coefficient in the SOP as follows
\begin{equation}
q_{n}(u)=u^{n}+k_{n}u^{n-1}+\mathcal{O}(u^{n-2}).
\end{equation}
Next, we expand the definition of the pre-kernel \eqref{eq:prek-def} to next-to-leading order for one argument being large, $\abs{u}\gg1$:
\begin{equation}
\prekernel_{n+1}(z,u) \sim \frac{1}{r_n} \left[-q_{2n}(z)u^{2n+1}+(q_{2n+1}(z)-k_{2n+1}q_{2n}(z))u^{2n}+\mathcal{O}(u^{2n-1})\right].
\end{equation}
It follows that the odd polynomials can be obtained as
\begin{equation}
q_{2n+1}(z)-k_{2n+1}q_{2n}(z) = \lim_{\abs{u}\to\infty} \frac{r_n\prekernel_{n+1}(z,u)+q_{2n}(z)u^{2n+1}}{u^{2n}},
\end{equation}
and likewise for the perturbed pre-kernel and SOP. Inserting the known expressions for the skew-norm, pre-kernel and even SOP of the perturbed weight on the right hand side, we thus obtain
\begin{equation}
\begin{split}
	&q_{2n+1}^{(1)}(z)-k_{2n+1}^{(1)}q_{2n}^{(1)}(z)
	=\lim_{\abs{u}\to\infty} \frac{r_n^{(1)}\prekernel_{n+1}^{(1)}(z,u)+q_{2n}^{(1)}(z)u^{2n+1}}{u^{2n}} \nonumber\\
	&\quad =\frac{1}{(z-m)q_{2n}(m)} \brackets[\Big]{q_{2n}(m)q_{2n+2}(z)-q_{2n}(z)q_{2n+2}(m)}
		-(k_{2n+2}+m) \frac{r_n\prekernel_{n + 1}(z,m)}{ (z-m )q_{2 n}(m)}.
\end{split}
\end{equation}
Recognising that the last term in the last line is just the perturbed even SOP $q_{2n}^{(1)}(z)$, this yields the expression for the odd perturbed SOP in the second equation of \eqref{SOP1}, with the constant given by $d_n=k_{2n+1}^{(1)}-k_{2n+2}-m$ here. Because this constant is arbitrary, we didn't specify it in \eqref{SOP1}.
\end{proof}
In principle, both even and odd  perturbed SOP $q_n^{(1)}(z)$ could be expanded in the basis of the perturbed OP $p_n^{(1)}(z)$. However, these are not as simple as in \autoref{thm:CP_SOP} and include the full sum of even and odd polynomials down to lowest order. We refer to \autoref{appSOP1} for details.



\begin{appendix}
\section{Recollection of known planar OP and SOP}\label{appA}

In this appendix we collect more planar OP and SOP.
In part they are already known, but for completeness (and because we use some of them in the main text) we state them here in as much generality as possible.
In particular, we can rederive the planar SOP from \autoref{thm:sop_from_op}. In contrast to the main text, in this appendix  we will consider the flat Lebesgue measure $\dif^2z=\dif x\dif y$ for $z=x+iy$.

\begin{example}[Product of Ginibre matrices]
When taking the product of $M$ complex Ginibre matrices we obtain for the weight function and norms \cite{ABu}
\begin{align}
	w(z) &= \abs{z}^{2 c}
		G_{0, M}^{M, 0}\left\lparen
		\begin{matrix}
			\raisebox{.5ex}{\rule{3em}{.4pt}}\\
			0, \dots, 0
		\end{matrix}
		\;\middle\lvert\;
		\abs{z}^2
		\right\rparen
	= G_{0, M}^{M, 0}\left\lparen
		\begin{matrix}
			\raisebox{.5ex}{\rule{3em}{.4pt}} \\
			c, \dots, c
		\end{matrix}
		\;\middle\lvert\;
		\abs{z}^2
		\right\rparen, \\
	h_n &= \pi \, \Gamma\p{n + 1 + c}^M.
\end{align}
Here, $G_{0, M}^{M, 0}$ is the Meijer $G$-function, see \cite[Chapter 16.17]{NIST} for the definition, and we have slightly extended \cite{ABu} by the insertion of a point charge $c > -1$ at the origin.
From \autoref{thm:sop_from_radial_weight} we have for the SOP and their skew-norms (also stated in \cite{Jesper})
\begin{equation}\begin{split}
	q_{2 k}(z) &= 
	\sum_{j = 0}^{k} z^{2 j} \prod_{l=j}^{k-1}(2l+3+c)^M 
		 = \p{2^k \Gamma\p{k + 1 + \frac{c}{2}}}^M \sum_{j = 0}^{k} \frac{z^{2 j}}{\p{2^j \Gamma\p{j + 1 + \frac{c}{2}}}^M}, \\
	q_{2 k + 1}(z) &= z^{2 k + 1}, \\
	r_k &= 2 \pi \Gamma\p{2 k + c + 2}^M.
\end{split}\end{equation}
Two particular cases are worth mentioning.
For $M=1$, that is a single Ginibre matrix, it holds that 
$G^{1\,0}_{0\,1}\left(\mbox{}_{c}^{-} \big| \ |z|^2 \right)
		=|z|^{2c}\euler^{-|z|^2}$ and we are back to the weight of the induced Ginibre ensemble, compare \autoref{ex:ML} at $\lambda=1$.
For $M=2$ 	we obtain $G^{2\,0}_{0\,2}\left(\mbox{}_{c,c}^{--} \big| \ |z|^2 \right)
		 =2|z|^{2c}K_0(2|z|)$, given in terms of the modified Bessel-function of the second kind $K_\nu$. At $c=0$ this corresponds to the weight of the chiral symplectic Ginibre ensemble at maximal non-Hermiticity, compare 
		 \cite{A05} where the corresponding SOP were constructed. 
\end{example}


\begin{example}[Elliptic Ginibre ensemble] \label{ex-HSOP}
The weight function of the {elliptic Ginibre ensemble} with parameters $A > B > 0$ is given by the complex normal distribution
\begin{equation}
	w(z) = \euler^{-A \abs{z}^2 + B \Re(z^2)}= \euler^{-(A-B)\Re(z)^2-(A+B)\Im(z)^2}.
\end{equation}
The monic OP $p_n(z)$, the recurrence coefficients $c_n$ and the squared norms $h_n$ are given by
\begin{equation}\label{eq:planarH}
\begin{split}
	p_n(z) &= \frac{1}{(2 C)^n} H_n\p{C z}, \quad C = \sqrt{\frac{A^2 - B^2}{2 B}}, \\
	c_n&=\frac{nB}{A^2-B^2},\\
	h_n &= \frac{\pi n!}{\sqrt{A^2 - B^2}} \p{\frac{A}{A^2 - B^2}}^n.
\end{split}\end{equation}
From \autoref{thm:sop_from_op} we get for the SOP
\begin{equation}\begin{split}
	q_{2 k}(z) &= k! \p{\frac{2 A}{A^2 - B^2}}^k
		\sum_{j = 0}^{k} \p{\frac{B}{2A}}^j \frac{1}{2^j j!} H_{2 j}\p{C z}, \\
	q_{2 k + 1}(z) &= \frac{1}{(2 C)^{2 k + 1}} H_{2 k + 1}\p{C z} .
\end{split}\end{equation}
For the skew-norms we obtain
\begin{equation}
	r_k = \frac{2 \pi}{(A + B) \sqrt{A^2 - B^2}} (2 k + 1)! \p{\frac{A}{A^2 - B^2}}^{2 k}.
\end{equation}
\end{example}
These polynomials reduce to the OP from \cite{FKS98} and to the SOP from  \cite{Kanzieper} when choosing  $A=1/(1-\tau^2)$ and $B=\tau/(1-\tau^2)$ for $0\leq\tau<1$.
The orthogonality relation \eqref{eq:HermiteOP} was proven first in \cite{EM} and independently in \cite{PdF}.
For self-consistency we present a proof for the well-known form of Hermite polynomials depending on two parameters \eqref{eq:planarH} following \cite{KS}.

\begin{proof}
Inserting the following integral representation \cite[Table 18.10.1]{NIST}
\begin{equation}
H_n(z)=n!\oint \euler^{2zt-t^2}t^{-n-1}\frac{\dif t}{2\pi\iunit},
\end{equation}
into the orthogonality relation, where the contour integral is around the origin in positive direction,  we obtain
\begin{align}
\MoveEqLeft 
\int_{\Cset}H_n(Cz)H_m(C\cconj{z})w(z)\dif^2z
\nonumber\\
&
=\int_{-\infty}^{\infty} \int_{-\infty}^{\infty}
\euler^{-(A-B)x^2-(A+B)y^2}
\oint\oint\frac{n!m!}{(2\pi\iunit)^2}\frac{e^{-t^2-s^2+2C(x+\iunit y)t+2C(x-\iunit y)s}}{t^{n+1}s^{m+1}}\dif t\dif s
\dif x \dif y \nonumber\\
&=\frac{\pi}{\sqrt{A^2-B^2}}\oint\oint\frac{n!m!}{(2\pi\iunit)^2}\frac{e^{\frac{2A}{B}st}}{t^{n+1}s^{m+1}}
\dif t\dif s\nonumber\\
&=\frac{\pi m!}{\sqrt{A^2-B^2}}\left(\frac{2A}{B}\right)^n\oint
\frac{1}{s^{m-n+1}}\frac{\dif s}{2\pi\iunit}.
\end{align}
Because the integrals are absolutely convergent, the order of integration can be interchanged, and using \eqref{eq:Gauss1} we have performed the two real integrations. Cauchy's integral theorem for derivatives leads to a single integral that gives $\delta_{n,m}$. Making the Hermite polynomials $H_n(x)=2^nx^2+\mathcal{O}(z^{n-1})$ monic, and using the known recurrence relation for Hermite we arrive at \eqref{eq:planarH}.
\end{proof}

\begin{example}[Chiral elliptic Ginibre ensemble] \label{ex-LSOP}
The weight function of the {chiral elliptic Ginibre ensemble} with parameters $A > B > 0$ and $\nu >-1$ is
\begin{equation}
	w(z) = \abs{z}^\nu K_{\nu}\p{A \abs{z}} \euler^{B \Re(z)}.
\end{equation}
The monic OP 
and their normalisation are given by
\begin{equation}\begin{split}
	p_n(z) &= \frac{(-1)^n n!}{C^n} L_n^{(\nu)}\p{C z}, \quad C = \frac{A^2 - B^2}{2 B}, \\
	c_n &= n (n + \nu) \p{\frac{2 B}{A^2 - B^2}}^2, \\
	h_n &= \frac{\pi}{A} n!\,  \Gamma(n + \nu+1) \p{\frac{2 A}{A^2 - B^2}}^{2 n + \nu + 1}.
\end{split}\end{equation}
Here $L_n^{(\nu)}$ denotes the $n$-th generalised Laguerre polynomial.
From \autoref{thm:sop_from_op} we obtain for the SOP 
and their skew-norms:
\begin{equation}\begin{split}
	q_{2 k}(z) &= 2^{2 k} k! \Gamma\p{k + \frac{\nu}{2} + 1} \p{\frac{2 A}{A^2 - B^2}}^{2 k}
		\sum_{j = 0}^{k} \p{\frac{B}{A}}^{2 j} \frac{(2 j)!}{2^{2 j} j! \Gamma\p{j + \frac{\nu}{2} + 1}} L_{2 j}^{(\nu)}\p{C z}, \\
	q_{2 k + 1}(z) &= -\frac{(2 k + 1)!}{C^{2 k + 1}} L_{2 k + 1}^{(\nu)}\p{C z}, \\
	r_k &= \frac{4 \pi}{A^2} (2 k + 1)!\, \Gamma(2 k + \nu + 2) \p{\frac{2 A}{A^2 - B^2}}^{4 k + \nu + 2}.
\end{split}\end{equation}
\end{example}
We use this ensemble in Subsection \ref{sec:PK-Laguerre} with the convention 
$A = 2 / (1 - \tau^2)$ and $B = 2 \tau / (1 - \tau^2)$ for $\tau \in [0, 1)$.
Note that the OP for this weight appeared in \cite{James} and the SOP were derived in \cite{A05} (in terms of squared eigenvalues) with $A = N (1 + \mu^2) / (2 \mu^2)$ and $B = N (1 - \mu^2) / (2 \mu^2)$ for $\mu \in (0, 1)$. For the orthogonality proof we refer to \cite{Karp,A05}.

\section{Fourier coefficients of the perturbed SOP}\label{appSOP1}

In this appendix we compute the expansion of the perturbed SOP $q_n^{(1)}$ from \autoref{thm:CP_SOP}, in the basis of the perturbed OP $p_n^{(1)}$ from \autoref{cor:OP1}, which are skew-orthogonal respectively orthogonal with respect to the perturbed weight $w^{(1)}(z)=|z-m|^2 w(z)$. Furthermore, we assume that the unperturbed OP $p_n$ obey a three-term recurrence relation, and thus \autoref{thm:sop_from_op} applies to determine the $q_n$.
As a result we will see that both even and odd polynomials $q_n^{(1)}$ have Fourier coefficients in even and odd degree of $p_n^{(1)}$, down to the lowest degree.

We begin with the polynomials of odd degree, defining the coefficients 
$\beta_{2k+1,j}$ as 
\begin{equation}\label{betadef}
q_{2k+1}^{(1)}(z)= \sum_{l=0}^{2k+1}\beta_{2k+1,l}p_l^{(1)}(z).
\end{equation}
It follows from the fact that both perturbed SOP and OP are monic, that 
$\beta_{2k+1,2k+1}=1$. Following  the definition we have for the remaining coefficients, with $l<2k+1$, that
\begin{align}
\beta_{2k+1,l} &= \frac{1}{h_l^{(1)}}\int_D q_{2k+1}^{(1)}(z)p_l^{(1)}(\cconj{z}) |z-m|^2w(z)\dif^2z
\nonumber\\
&= \frac{1}{h_l^{(1)}}\int_D\frac{q_{2k+2}(m)\sum_{j=0}^k\mu_{k,j}p_{2j}(z)-
q_{2k}(m)\sum_{j=0}^{k+1}\mu_{k+1,j}p_{2j}(z)}{q_{2k}(m)}
\nonumber\\
&\hphantom{\frac{1}{h_l^{(1)}}\int_D} \quad \times
\frac{K_{l+1}(\cconj{z},m)p_{l+1}(m)-K_{l+1}(m,m)p_{l+1}(\cconj{z})}{K_{l+1}(m,m)}
w(z)\dif^2z
\nonumber\\
&= \frac{1}{h_{l+1}K_{l+2}(m,m)q_{2k}(m)}
\left[\sum_{j=0}^{\floor{l/2}}\left(q_{2k+2}(m)\mu_{k,j}-q_{2k}(m)\mu_{k+1,j}\right)p_{2j}(m)p_{l+1}(m)
\right.\nonumber\\
&\quad\quad \quad\quad
-\delta_{l,2L+1}(q_{2k+2}(m)\mu_{k,L+1}-q_{2k}(m)\mu_{k+1,L+1})
K_{2L+2}(m,m)h_{2L+2}
\Bigg].
\label{betaresult}
\end{align}
In the first step we have inserted the perturbed SOP and OP from eqs.~\eqref{SOP1} and \eqref{pertOP}, leading to a cancellation of $|z-m|^2$. Furthermore, for simplicity we have set $d_k=0$ in the former (otherwise this would contribute  to the coefficients $\alpha_{2k,j}$ from \eqref{alphadef} below). Next, we have used the orthogonality of the unperturbed OP, as well as the projection property of the kernel, 
\begin{equation}
\int_D p_{j}(z) K_{l+1}(\cconj{z},m) w(z) \dif^2z= p_j(m), \quad \mbox{for}\ \  j\leq l,
\end{equation}
and zero otherwise. In the final result \eqref{betaresult} the last term is non-vanishing only when $l=2L+1$ is odd, whereas the previous term is also present for $l=2L$ even (the sum runs to $\floor{l/2}=L$ in both cases). Consequently, all even and odd Fourier coefficients $\beta_{2k+1,l}$ are non-vanishing in general, down to the lowest degree $l=0$, in contrast to \autoref{thm:sop_from_op}. 

Let us move to the Fourier coefficients of the even polynomials, defined as 
\begin{equation}\label{alphadef}
q_{2k}^{(1)}(z)= \sum_{l=0}^{2k}\alpha_{2k,l}p_l^{(1)}(z)\ .
\end{equation}
As for the odd polynomials, we have from the monic property of the two sets of polynomials that $\alpha_{2k,2k}=1$. For the remaining coefficients with $l<2k$ we obtain
\begin{align}
\alpha_{2k,l} &= \frac{1}{h_l^{(1)}}\int_D q_{2k}^{(1)}(z)p_l^{(1)}(\cconj{z}) |z-m|^2w(z)\dif^2z
\nonumber\\
&= \frac{r_k}{h_l^{(1)}}\int_D\frac{\sum_{i=0}^{k}\frac{1}{r_i} \left(q_{2i+1}(m)\sum_{j=0}^i\mu_{i,j}p_{2j}(z)-
q_{2i}(m)p_{2i+1}(z)\right)}{q_{2k}(m)}
\nonumber\\
&\hphantom{\frac{r_k}{h_l^{(1)}}\int_D} \quad \times
\frac{K_{l+1}(\cconj{z},m)p_{l+1}(m)-K_{l+1}(m,m)p_{l+1}(\cconj{z})}{K_{l+1}(m,m)}
w(z)\dif^2z
\nonumber\\
&=\frac{r_k}{h_{l+1}K_{l+2}(m,m)q_{2k}(m)}
\int_D\left(\sum_{j=0}^k\sum_{i=j}^{k}\frac{1}{r_i} q_{2i+1}(m)\mu_{i,j}p_{2j}(z)-
\sum_{i=0}^{k}\frac{1}{r_i} q_{2i}(m)p_{2i+1}(z)\right)
\nonumber\\
&\hphantom{\frac{r_k}{h_{l+1}K_{l+2}(m,m)q_{2k}(m)}\int_D} \quad \times
\left(K_{l+1}(\cconj{z},m)p_{l+1}(m)-K_{l+1}(m,m)p_{l+1}(\cconj{z})\right)
w(z)\dif^2z,
\end{align}
where we have followed the same procedure as before, and swapped the summation in the double sum in the last step, to facilitate the integration.
Let us distinguish even and odd indices $l$ now. For even $l=2L$, with $L<k$ we obtain
\begin{align}
\alpha_{2k,2L} &= \frac{r_k}{h_{2L+1}K_{2L+2}(m,m)q_{2k}(m)}
\left[\sum_{j=0}^L\sum_{i=j}^{k}\frac{1}{r_i}q_{2i+1}(m)\mu_{i,j}p_{2j}(m) p_{2L+1}(m)
\right.
\nonumber\\
&\qquad
-\sum_{i=0}^{L-1}\frac{1}{r_i}q_{2i}(m)p_{2i+1}(m) p_{2L+1}(m)
 +\frac{1}{r_L} q_{2L}(m)K_{2L+1}(m,m)h_{2L+1}
\Bigg].\ \ \ 
\label{alphaeven}
\end{align}
For $l=2L+1$ odd with $L<k$ it follows that 
\begin{align}
\alpha_{2k,2L+1}&=\frac{r_k}{h_{2L+2}K_{2L+3}(m,m)q_{2k}(m)}
\left[\sum_{j=0}^{L}\sum_{i=j}^{k}\frac{1}{r_i} q_{2i+1}(m)\mu_{i,j}p_{2j}(m)p_{2L+2}(m)
\right.
\nonumber\\
&\quad
-\sum_{i=0}^{L}\frac{1}{r_i} q_{2i}(m)p_{2i+1}(m)p_{2L+2}(m)
-\sum_{i=L+1}^{k}\frac{1}{r_i}q_{2i+1}(m)\mu_{i,L+1} K_{2L+2}(m,m)h_{2L+2}
\Bigg].
\label{alphaodd}
\end{align}
Once again all even and odd coefficients $\alpha_{2k,l}$ are non-vanishing, down to the lowest degree $l=0$.

\section{Some useful integrals}\label{appC}

For completeness we collect a few simple Gaussian integrals that will be useful in several places throughout the main part.
\begin{enumerate}
	\item For all $\alpha > 0$ and $\beta \in \Cset$ it holds:
	\begin{equation}\label{eq:Gauss1}
		\int_{-\infty}^{\infty} \euler^{-\alpha t^2 + \beta t} \dif t
		= \sqrt{\frac{\pi}{\alpha}} \exp\p{\frac{\beta^2}{4 \alpha}}.
	\end{equation}
	
		\item For all $\alpha > 0, \beta,\gamma,\delta \in \Cset$  with $\Re(\alpha + \gamma^2) > 0$ it holds, compare \cite[8.259.1]{Gradshteyn} :
	\begin{equation}\label{eq:Gauss2}
		\int_{-\infty}^{\infty} 
		\euler^{-\alpha t^2 + \beta t}
		\erf\p{\gamma t+\delta}
		\dif t
		= \sqrt{\frac{\pi}{\alpha}}
		\exp\p{\frac{\beta^2}{4 \alpha}}
		\erf\p{\frac{\beta\gamma+2\alpha\delta}{2\sqrt{\alpha(\alpha + \gamma^2)}}}.
	\end{equation}
	
	\item Applying \eqref{eq:Gauss2} twice, it follows that for all $A, B > 0$, $C, D, \zeta, \eta \in \Cset$ it holds:
	\begin{equation}\label{eq:Gauss3}
		\int_{-\infty}^{\infty} \int_{-\infty}^{\infty}
		\euler^{- A t^2 - B s^2 + 2 \iunit (t \zeta + s \eta)}
		\erf\p{C t + D s} \dif s \dif t 
		= \frac{\pi\euler^{- \frac{\zeta^2}{A} - \frac{\eta^2}{B}}}{\sqrt{A B}}
		\erf\p{\iunit \frac{B C \zeta + A D \eta}{\sqrt{A B \p{A B + A D^2 + B C^2}}}}.
	\end{equation}
\end{enumerate}

In Subsection \ref{sec:PK-Laguerre} we encounter the following integrals involving Bessel-$J$ and Bessel-$I$ functions.
\begin{enumerate}
	\item For $\Re(\nu) > -1$, $u \in \Cset$ it holds:
	\begin{equation} \label{eq:Bessel Gamma integrals}
	\begin{split}
		\int_{0}^{\infty} \euler^{-q^2} \Gamma\p{\frac{\nu + 1}{2}, q^2} J_{\nu}\p{2 q \sqrt{2 u}} \dif q
		&= \frac{\sqrt{\pi}}{2} \Gamma\p{\frac{\nu + 1}{2}} \euler^{-u} I_{\nu / 2}(u) \\
		&\quad- \frac{1}{2} \p{\frac{u}{2}}^{\nu / 2} \int_{-1}^{0} \p{(1 - t) (1 + t)}^{\nu / 2 - 1/2} \euler^{-u (1 - t)} \dif t, \\
		\int_{0}^{\infty} \euler^{-q^2} \Gamma\p{\frac{\nu + 1}{2}, q^2} I_{\nu}\p{2 q \sqrt{2 u}} \dif q
		&= \frac{\sqrt{\pi}}{2} \Gamma\p{\frac{\nu + 1}{2}} \euler^{u} I_{\nu / 2}(u) \\
		&\quad - \frac{1}{2} \p{\frac{u}{2}}^{\nu / 2} \int_{-1}^{0} \p{(1 - t) (1 + t)}^{\nu / 2 - 1/2} \euler^{u (1 - t)} \dif t.
	\end{split}
	\end{equation}
	
	\item For $\Re(\nu) > -1$, $\Re(c) > 0$ it holds (compare \cite[10.22.67]{NIST} and \cite[10.43.28]{NIST}):
	\begin{equation} \label{eq:Bessel product integrals}
	\begin{split}
	\int_{0}^{\infty} t \euler^{-c t^2} J_\nu(a t) J_\nu(b t) \dif t
		&= \frac{1}{2 c} \exp\p{\frac{-a^2 - b^2}{4 c}} I_\nu\p{\frac{a b}{2 c}}, \\
	\int_{0}^{\infty} t \euler^{-c t^2} J_\nu(a t) I_\nu(b t) \dif t
		&= \frac{1}{2 c} \exp\p{\frac{-a^2 + b^2}{4 c}} J_\nu\p{\frac{a b}{2 c}}, \\
	\int_{0}^{\infty} t \euler^{-c t^2} I_\nu(a t) I_\nu(b t) \dif t
		&= \frac{1}{2 c} \exp\p{\frac{a^2 + b^2}{4 c}} I_\nu\p{\frac{a b}{2 c}}.
	\end{split}\end{equation}
	All three formulas are equivalent because $I_\nu(\iunit z) = \iunit^\nu J_\nu(z)$.
	
	\item For $\Re(\nu) > -1$, $u, v \in \Cset$ it holds:
	\begin{equation} 
	\label{eq:Bessel product double-integrals}
	\begin{split}
		\euler^{u + v} \int_{0}^{\infty} \int_{0}^{q} \euler^{-q^2 - p^2} J_{\nu}\p{2 p \sqrt{2 u}} J_{\nu}\p{2 q \sqrt{2 v}} \dif p \dif q
		&= \frac{1}{4} \int_{0}^{\pi / 2} \euler^{(u - v) \cos(\alpha)} I_{\nu}\p{2 \sqrt{u v} \sin(\alpha)} \dif \alpha, \quad\\
		\euler^{u - v} \int_{0}^{\infty} \int_{0}^{q} \euler^{-q^2 - p^2} J_{\nu}\p{2 p \sqrt{2 u}} I_{\nu}\p{2 q \sqrt{2 v}} \dif p \dif q
		&= \frac{1}{4} \int_{0}^{\pi / 2} \euler^{(u + v) \cos(\alpha)} J_{\nu}\p{2 \sqrt{u v} \sin(\alpha)} \dif \alpha.
	\end{split}\end{equation}
\end{enumerate}

\begin{proof}
For \eqref{eq:Bessel Gamma integrals} we write the incomplete Gamma function as
\begin{equation}
	\Gamma\p{\frac{\nu + 1}{2}, q^2}
	= 2 \p{\int_{0}^{\infty} t^\nu \euler^{-t^2} \dif t - \int_{0}^{q} t^\nu \euler^{-t^2} \dif t}
	= \Gamma\p{\frac{\nu + 1}{2}} - 2 q^{\nu + 1} \int_{0}^{1} s^\nu \euler^{-q^2 s^2} \dif s.
\end{equation}
Now we can compute the integral over $q$, with \cite[10.22.52]{NIST} and \cite[10.22.51]{NIST} for Bessel-$J$, and \cite[10.43.24]{NIST} and \cite[10.43.23]{NIST} for Bessel-$I$.
The remaining $s$-integral can be simplified with the substitution $s = \sqrt{(1 + t) / (1 - t)}$.

For \eqref{eq:Bessel product double-integrals} we first make the substitution $p = q r$, $r \in [0, 1]$, then we can switch the integrals and evaluate the $q$-integral with \eqref{eq:Bessel product integrals}.
For the first integral we arrive at
\begin{equation}
	\frac{1}{2} \int_{0}^{1} \frac{1}{1 + r^2} \exp\p{\frac{1 - r^2}{1 + r^2} (u - v)} I_\nu\p{\frac{4 r}{1 + r^2} \sqrt{u v}} \dif r,
\end{equation}
the second integral is analogous with $(u + v)$ in the exponential and a Bessel-$J$ function instead of Bessel-$I$.
The substitution $r = \tan(\alpha / 2)$, $\alpha \in [0, \pi / 2]$, gets rid of the factor  
$1 / (1 + r^2)$.
\end{proof}

\end{appendix}



\begin{thebibliography}{99}

\small

\bibitem{Ginibre}
J. Ginibre, 
Statistical ensembles of complex, quaternion, and real matrices.
J. Math. Phys. {\bf 6} (1965) 440. 

\bibitem{LehmannSommers}
N. Lehmann and H.-J. Sommers, 
Eigenvalue statistics of random real matrices
Phys. Rev. Lett. {\bf 67} (1991)  941.


\bibitem{Edelman}
A. Edelman,
The probability that a random real gaussian matrix has k real eigenvalues, related distributions, and the circular law.
J. Multivar. Anal. {\bf 60} (1997) 203.


\bibitem{James}
J.C. Osborn, 
Universal results from an alternate random-matrix model for QCD with a baryon chemical potential. 
Phys. Rev. Lett. {\bf 93} (2004) 222001 [hep-th/0403131].


\bibitem{A05}
G. Akemann,
{The Complex Laguerre Symplectic Ensemble of Non-Hermitian Matrices}.
Nucl. Phys. B {\bf 730}
(2005) 253 
[hep-th/0507156].


\bibitem{APS}
G. Akemann, M.J. Phillips, and H.-J. Sommers, 
The chiral Gaussian two-matrix ensemble of real asymmetric matrices. 
J. Phys. A: Math. Theor.  {\bf 43} 
(2010) 085211 [arXiv:0911.1276].


\bibitem{KSZ}
B.A. Khoruzhenko, H.-J. Sommers, and K. \.Zyczkowski, Truncations of random orthogonal matrices. Phys. Rev. E {\bf 82} 
(2010) 040106 
[arXiv:1008.2075].

\bibitem{ZS}
K. \.Zyczkowski and H.-J.  Sommers,  
Truncations of random unitary matrices. 
J. Phys. A: Math. Gen. {\bf 33}
(2000) 2045 
[chao-dyn/9910032].

\bibitem{BL} Boris A. Khoruzhenko and Serhii Lysychkin,
Truncations of random unitary symplectic matrices, in preparation.



\bibitem{KS} B.A. Khoruzhenko and H.-J. Sommers, {Non-Hermitian Random Matrix Ensembles}. Chapter 18 in G.~Akemann, J.~Baik and  P.~Di Francesco (eds.), 
{The Oxford Handbook of Random Matrix Theory}, Oxford Univiversity Press 2011, Oxford [arXiv:0911.5645].

\bibitem{Peter}
P.J. Forrester,
Log-Gases and Random Matrices. 
Princeton University Press 2010, Princeton.


\bibitem{EMeckes}
E. Meckes, The Random Matrix Theory of the Classical Compact Groups (Cambridge Tracts in Mathematics).  Cambridge University Press 2019, Cambridge.


\bibitem{Haake}
F. Haake, Quantum Signatures of Chaos. Springer, 3rd Edition 2010, Heidelberg.

\bibitem{Crisanti}
H. Sompolinsky, A. Crisanti, and H. J. Sommers, 
Chaos in Random Neural Networks. Phys. Rev. Lett. {\bf 61} (1988) 259.

\bibitem{PdF} P. Di Francesco, M.  Gaudin, C. Itzykson and F. Lesage, 
Laughlin's wave functions, Coulomb gases and expansions of the discriminant.
Int. J. Mod. Phys. A {\bf 9} (1994) 4257 
[hep-th/9401163].	


\bibitem{ABG} 
K. G\'orska, A. Horzela and F.H. Szafraniec,
Coherence, Squeezing and Entanglement: An Example of
Peaceful Coexistence. Chapter 5 in  J.-P. Antoine, F. Bagarello
J.-P. Gazeau (eds.), Coherent States and Their
Applications, Springer 2018, Cham.

\bibitem{EK} A.V. Kolesnikov and K.B. Efetov, 
{Distribution of complex eigenvalues for symplectic ensembles of non-Hermitian matrices}. Waves Rand. Media \textbf{9} (1999) 71 
[cond-mat/9809173].

\bibitem{DBB} J.P. Dahlhaus, B. B\'eri, and C.W.J. Beenakker, Random-matrix theory of thermal conduction in superconducting quantum dots. Phys. Rev. B {\bf 82} (2010) 014536 [arXiv:1004.2438].

\bibitem{ABi} G. Akemann and E. Bittner, Unquenched complex Dirac spectra at nonzero chemical potential: two-colour QCD lattice data versus matrix model. Phys. Rev. Lett. \textbf{96} (2006)  222002 
[hep-lat/0603004].


\bibitem{Mehta}
M.L. Mehta, {Random {M}atrices}.
Academic Press, 2nd Edition 1990, New York.


\bibitem{Kanzieper}
E.~Kanzieper, {Eigenvalue correlations in non-{H}ermitean symplectic random matrices}.  
J. Phys. A: Math. Gen. \textbf{35} (2002)  6631 
[cond-mat/0109287].


\bibitem{Rider}
B. Rider, {A limit theorem at the edge of a non-Hermitian random matrix ensemble}. 
J. Phys. A: Math. Gen. {\bf 36}  
(2003) 3401. 

\bibitem{Jesper}
J.R. Ipsen, {Products of Independent Quaternion Ginibre Matrices and their Correlation Functions}.
J. Phys. A: Math. Theor. {\bf 46} 
(2013) 265201 
[arxiv:1301.3343].



\bibitem{AKMP}
G. Akemann, M. Kieburg, A. Mielke, and T. Prosen, 
Universal signature from integrability to chaos in dissipative open
quantum systems. Phys.~Rev.~Lett. {\bf 123} (2019) 254101 
[arXiv:1910.03520].

\bibitem{BS}
A. Borodin and C.D. Sinclair,  
The Ginibre ensemble of real random matrices and its scaling limits.
Commun. Math. Phys. {\bf 291} (2009) 177 
[arXiv:0805.2986].



\bibitem{Petersymplectic}
P.J. Forrester, {Skew orthogonal polynomials for the real and quaternion real Ginibre ensembles and generalizations}.
J. Phys. A: Math. Theor. {\bf 46} 
(2013) 245203 
[arxiv:1302.2638].

\bibitem{AKP}
G. Akemann, M. Kieburg, and M.J. Phillips. 
Skew-orthogonal Laguerre polynomials for chiral real asymmetric random matrices.
J. Phys. A: Math. Theor. {\bf 43} (2010) 375207 
[arXiv:1005.2983].



\bibitem{Lempert}
L. Lempert, Recursion for orthogonal polynomials on complex domains. In: Fourier Analysis and
Approximation Theory (Proc. Colloq., Budapest, 1976), Vol. II, pp. 481--494, North-Holland 1978, Amsterdam.


\bibitem{Khavinson}
D. Khavinson and  N. Stylianopoulos. 
{Recurrence Relations for Orthogonal Polynomials and Algebraicity of Solutions of the Dirichlet Problem}. In: A. Laptev (ed.) Around the Research of Vladimir Maz'ya II. International Mathematical Series, vol 12. Springer 2010, New York.

\bibitem{ANPV}
G. Akemann, T. Nagao, I. Parra, and G. Vernizzi, {Gegenbauer and other planar orthogonal polynomials on an ellipse in the complex plane}.
Constr. Approx. {\bf 53} (2021) 441 
[arxiv:1905.02397].

\bibitem{AHM}
M. Adler, E. Horozov, and P. van Moerbeke, The Pfaff lattice and skew-orthogonal polynomials. Int. Math. Res. Notices {\bf 11} (1999) 569 
[solv-int/9903005].


\bibitem{Walter}
W. Van Assche, Orthogonal polynomials in the complex plane and on the real line. In: Special functions, q-series and related topics, M.E.H. Ismail et al. (eds), Fields Inst. Commun. {\bf 14} (1997) 211--245.


\bibitem{AFNM}
M. Adler, P. J. Forrester, T. Nagao, P. van Moerbeke.
Classical Skew Orthogonal Polynomials and Random Matrices.
J. Stat. Phys. {\bf 99}  (2000) 141 
[arXiv:solv-int/9907001].


\bibitem{AKS}
Y. Ameur, N.-G. Kang, and S.-M. Seo, The random normal matrix model: Insertion of a point charge. 
Potential Anal (2021) online 
[arXiv:1804.08587]


\bibitem{Fischman}
J. Fischmann, W. Bruzda, B.A. Khoruzhenko, H.-J. Sommers, and K. \.Zyczkowski,  {Induced Ginibre ensemble of random matrices and quantum operations}. 
J. Phys. A: Math. Theor. {\bf 45} (2012)  075203 [arXiv:1107.5019]. 



\bibitem{Mason} J.C. Mason and D.C. Handscomb, Chebyshev Polynomials. 
Chapman and Hall/CRC 2002, London.



\bibitem{FKS98} 
Y.V.\ Fyodorov, B.A.\ Khoruzhenko and H.-J.\ Sommers, 
Universality in the random matrix spectra in the regime of weak non-Hermiticity.
Ann.\ Inst.\ Henri Poincar\'e
{\bf 68} (1998) 449 
[chao-dyn/9802025].

\bibitem{EM}
S.J.L. van Eijndhoven and J.L.H. Meyers, New orthogonality relation for the Hermite polynomials and related Hilbert spaces. J. Math. Anal. Appl. {\bf 146} (1990) 89. 


\bibitem{NIST}
F.W.J.  Olver, D.W.  Lozier, R.~F. Boisvert and C.~W. Clark (eds.), 
{NIST Handbook of Mathematical Functions}.  
Cambridge University Press 2010, Cambridge.


\bibitem{H.Amann}
H. Amann,  Ordinary Differential Equations. De Gruyter 1990, Berlin.


\bibitem{VBarg}
V. Bargmann,  On a Hilbert Space of Analytic Functions
and an Associated Integral Transform
Part I.
Commun. Pure Appl. Math. {\bf XIV} (1961) 187. 


\bibitem{Ameur}
Y. Ameur, N.G.  Kang, N. Makarov, N., and A. Wennman, Scaling limits of random normal matrix processes at singular boundary points. J. Funct. Ana. {\bf 278}
(2020) 108340 [arXiv:1510.08723].


\bibitem{B-GC}
Florent Benaych-Georges and  Francois Chapon, 
Random right eigenvalues of Gaussian quaternionic matrices.
Rand. Matr.: Theor. Appl. {\bf 1} (2012) 1150009 
[arXiv:1104.4455]. 


\bibitem{HJS}
H. J. Sommers, A. Crisanti, H. Sompolinsky, and Y. Stein
Spectrum of Large Random Asymmetric Matrices.
Phys. Rev. Lett. {\bf 60} (1988) 1895. 

\bibitem{BE}
Sung-Soo Byun, Markus Ebke, 
Universal scaling limits of the symplectic elliptic Ginibre ensemble. 
arXiv preprint arXiv:2108.05541, 2021.


\bibitem{Karp}
D. Karp, Holomorphic Spaces Related to Orthogonal Polynomials and Analytic Continuation of
Functions. In: Analytic Extension Formulas and their Applications.
Vol. {\bf 9}  pp. 169--187, 
Springer 2001, Boston, MA.



\bibitem{FC} F.M. Cholewinski, Generalized Fock spaces and associated operators.
SIAM J. Math. Anal.
{\bf 15} 
(1984) 177. 

\bibitem{ABK}
G. Akemann, S.-S. Byun, and N.-G Kang.
A non-Hermitian generalisation of the Marchenko–Pastur distribution: From the circular law to multi-criticality. 
Ann. Henri Poincar\'e {\bf 22} 
(2021) 1035 
[arXiv:2004.07626].


\bibitem{BDS}
J. Baik, P. Deift, E. Strahov,
Products and ratios of characteristic polynomials of random Hermitian matrices.
J. Math. Phys. {\bf 44} (2003) 3657 [arXiv:math-ph/0304016].



\bibitem{AV}
G. Akemann and G. Vernizzi,
Characteristic polynomials of complex random matrix models.
Nucl. Phys. B {\bf 660} (2003)  532 
[hep-th/0212051].

\bibitem{ABa}
G. Akemann and  F. Basile,
Massive partition functions and complex eigenvalue correlations in matrix models with symplectic symmetry.
Nucl. Phys. B {\bf 766} (2007) 150 
[math-ph/0606060].



\bibitem{ABu}
G. Akemann, Z. Burda, 
Universal microscopic correlation functions for products of independent Ginibre matrices.
J. Phys. A: Math. Theor. {\bf 45} (2012) 465201
[arXiv:1208.0187].


\bibitem{Gradshteyn} 
I.~S.~Gradshteyn and I.~M.~Ryzhik. Table of Integrals, Series, and Products, 
Academic Press, 7th edition 2007, San Diego.



\end{thebibliography}
\end{document}